\documentclass[11pt]{article}

\usepackage[utf8]{inputenc}
\usepackage[T1]{fontenc}

\usepackage{amsmath}
\usepackage{amsthm}
\usepackage{amssymb}
\usepackage{mathtools}

\usepackage{mathbbol} % Uncomment if you prefer \mathbb{} from mathbbol

\usepackage[compress]{natbib}

\usepackage[dvipsnames]{xcolor}
\usepackage{colortbl} % Removed 'color' as xcolor is loaded
\usepackage{graphicx}

\usepackage{booktabs}
\usepackage{multirow}
\usepackage{caption}
\usepackage{subcaption}

\usepackage{enumitem} % Enhanced list control

\usepackage{geometry}
\usepackage{changepage} % For adjustwidth environment

\usepackage[hidelinks]{hyperref} % Loaded after most packages
\usepackage{url}

\usepackage{mdframed}

\usepackage{setspace}

\usepackage{authblk}

\usepackage[linesnumbered,ruled,vlined]{algorithm2e}
\DontPrintSemicolon
\SetKwInOut{Input}{Input}
\SetKwInOut{Output}{Output}

\newcommand{\E}{\mathbb{E}}
\newcommand{\Prob}{\mathbb{P}}
\newcommand{\Cov}{\mathrm{Cov}}

\newcommand{\indep}{\rotatebox[origin=c]{90}{$\models$}}

\newcommand{\Lopt}{L^{*}}
\DeclareMathOperator*{\argmax}{arg\,max} % thin space, limits underneath in displays
\newcommand{\IF}[1]{\mathbb{IF}[#1]}

\newcommand{\ff}[2]{\frac{#1}{#2}}

\newcommand{\I}[1]{\mathbb{1}\!\left\{ #1 \right\}}

\newcommand{\ka}{\kappa}
\newcommand{\dy}[1]{\text{d}#1}
\newcommand{\pto}{\overset{P}{\to}}
\newcommand{\dto}{\overset{\mathcal{D}}{\to}}
\newcommand{\hka}{\hat\ka}
\newcommand{\hL}{\hat{L}}

\newcommand{\GP}{\mathcal{GP}}
\newcommand{\G}{G}
\newcommand{\CIh}{\text{CI}^{\text{H}}}
\newcommand{\tCIh}{\tilde{\text{CI}}^{\text{H}}}
\newcommand{\CIb}{\text{CI}^{\text{B}}}

\newcommand{\dka}{\ka^{\dagger}}
\newcommand{\dL}{{L}^{\dagger}}
\newcommand{\hdka}{\hat\ka^{\dagger}}
\newcommand{\hnddka}{\hat{\dot{\ka}}}
\newcommand{\hdL}{\hat{L}^{\dagger}}

\newcommand{\dv}{{\varrho}}
\newcommand{\ddv}{{\dot{\varrho}}} 
\newcommand{\hdv}{\hat{\varrho}}
\newcommand{\hddv}{\hat{\dot{\varrho}}} % first d for dagger section, second d for dot.
\newcommand{\ct}{\text{ct}}

\DeclarePairedDelimiter\ceil{\lceil}{\rceil}
\DeclarePairedDelimiter\floor{\lfloor}{\rfloor}
\newtheorem{theorem}{Theorem}
\newtheorem{proposition}{Proposition}
\newtheorem{lemma}{Lemma}
\newtheorem{assumption}{Assumption}
\newtheorem{remark}{Remark}

\newenvironment{assumbis}[1]
  {\renewcommand{\theassumption}{\ref{#1}$'$}%
   \addtocounter{assumption}{-1}%
   \begin{assumption}}
  {\end{assumption}}
\begin{document}

%% Here are the title, author names and addresses
\title{Beyond Fixed Restriction Time: Adaptive Restricted Mean Survival Time Methods in Clinical Trials}

\author[a]{Jinghao Sun\thanks{\texttt{jinghao.sun@pennmedicine.upenn.edu}}}
\author[a]{Douglas E. Schaubel\thanks{\texttt{douglas.schaubel@pennmedicine.upenn.edu}}}
\author[a,b]{Eric J. Tchetgen Tchetgen\thanks{\texttt{ett@wharton.upenn.edu}}}

\affil[a]{Department of Biostatistics, Epidemiology, \& Informatics, University of Pennsylvania}
\affil[b]{Department of Statistics and Data Science, University of Pennsylvania}

\maketitle

\begin{abstract}
Restricted mean survival time (RMST) offers a compelling nonparametric alternative to hazard ratios for right-censored time-to-event data, particularly when the proportional hazards assumption is violated. By capturing the total event-free time over a specified horizon, RMST provides an intuitive and clinically meaningful measure of absolute treatment benefit. Nonetheless, selecting the restriction time $L$ poses challenges: choosing a small $L$ can overlook late-emerging benefits, whereas a large $L$, often underestimated in its impact, may inflate variance and undermine power. We propose a novel data-driven, adaptive procedure that identifies the optimal restriction time $L^*$ from a continuous range by maximizing a criterion balancing effect size and estimation precision. Consequently, our procedure is particularly useful when the pattern of the treatment effect is unknown at the design stage. We provide a rigorous theoretical foundation that accounts for variability introduced by adaptively choosing $L^*$. To address nonregular estimation under the null, we develop two complementary strategies: a convex-hull-based estimator, and a penalized approach that further enhances power. Additionally, when restriction time candidates are defined on a discrete grid, we propose a procedure that surprisingly incurs no asymptotic penalty for selection, thus achieving oracle performance. Extensive simulations across realistic survival scenarios demonstrate that our method outperforms traditional RMST analyses and the log-rank test, achieving superior power while maintaining nominal Type I error rates. In a phase III pancreatic cancer trial with transient treatment effects, our procedure uncovers clinically meaningful benefits that standard methods overlook. Our methods are implemented in the \texttt{R} package \texttt{AdaRMST}.\\


\end{abstract}

\section{Introduction}

In randomized controlled trials evaluating new medications or procedures, right-censored time-to-event outcomes are crucial for assessing the efficacy of interventions over time. These outcomes, such as mortality, relapse, and disease progression, provide dynamic insights into treatment effects across different time points \citep{wilson2015outcomes}. Various statistical metrics are employed to analyze such data, including survival probabilities, hazard rates, median survival times, and Restricted Mean Survival Time (RMST) \citep{andersen1993statistical}.

Traditionally, the hazard ratio (HR) derived from the Cox proportional hazards model \citep{cox1972regression} has been the predominant measure for summarizing treatment effects in time-to-event analyses. The HR relies on the proportional hazards (PH) assumption, which posits that the hazard functions of the treatment groups remain proportional over time.
However, violations of this assumption can arise due to various mechanisms. For example, non-proportional hazards (non-PH) occur when the treatment effect changes over time, such as when transient benefits diminish as the study progresses, or when a delayed therapeutic effect is observed. For instance, in the context of COVID-19, vaccines often lead to an initial strong immune response that rapidly reduces the risk of infection or severe disease shortly after vaccination. However, over time, this protective effect can wane, especially against new variants of the virus \citep{lin2024durability}.  
In oncology, certain immunotherapies may take time to reach their full efficacy, leading to a delayed treatment effect \citep{ferris2016nivolumab, hoos2010improved}. When the proportional hazards assumption is violated—such as in the cases described above—the HR can become difficult to interpret, potentially misleading or uninformative \citep{hernan2010hazards, uno2014moving, stensrud2020test, vansteelandt2024assumption}.

RMST, by contrast, has emerged as a robust alternative that does not require the PH assumption \citep{royston2011use, zhang2012double, tian2018efficiency}. The RMST represents the mean event-free time up to a specified restriction time $L$, offering an intuitive and clinically meaningful measure of treatment efficacy. Its straightforward interpretation makes it appealing to both clinicians and patients \citep{royston2013restricted, pak2017interpretability}. 

A critical aspect of RMST analysis is the choice of the restriction time, $L$, as the RMST difference between treatment groups depends heavily on this parameter. Selecting an appropriate $L$ is essential for reliable inference and meaningful clinical interpretation. Domain knowledge, clinical relevance, and regulatory guidelines sometimes guide the choice of $L$ \citep{eaton2020designing}. However, in practice, such information may often be limited at the design stage, especially when the pattern of the treatment effect is uncertain. In such cases, \citet{hasegawa2020restricted} suggest selecting $L$ based on the maximum follow-up time ``so that the RMST approximates the mean survival time'', while \citet{tian2020empirical} discuss a data-dependent approach to selecting $L$ based on sample quantiles. e.g. 95th percentile of observed follow-up times.

In this paper, we propose a novel, data-driven methodology that builds on the idea that selecting $L$ involves a trade-off between capturing the treatment effect and maintaining statistical precision. On one hand, choosing a small $L$ may result in insufficient events or miss detecting a late-emerging treatment effect, leading to an underestimation of the treatment's efficacy. 
On the other hand, because the variance (and the commonly used variance estimators) of the RMST estimate are monotone increasing functions of $L$ as accumulating error compounds over time, choosing a large $L$ inevitably inflates the variance. The resulting wider confidence intervals may diminish statistical power and hinder detection of significant treatment differences. 

Our methodology seeks to balance this trade-off by optimally selecting the restriction time, $\Lopt$, from a set of candidate values. Specifically, we maximize a criterion that accounts for both the magnitude of the RMST difference and its variability. This is achieved by maximizing the ratio of the absolute RMST difference to its standard deviation, thereby improving the ability to detect meaningful treatment effects under various survival difference patterns. Consequently, our procedure is particularly useful when the pattern of the treatment effect is unknown at the design stage.
Moreover, our method provides a single-number summary of the treatment effect, which may be preferable for researchers seeking a concise summary rather than an entire survival curve. Specifically, it ensures coherence between hypothesis testing and estimation regardless of whether PH holds: the null hypothesis is rejected if and only if the confidence interval for the treatment effect excludes zero. 

Our methodology contributes to the literature in three significant ways: %
(1) We introduce a novel adaptive RMST estimation method that dynamically selects the optimal restriction time, $\Lopt$, from a \emph{continuous-time} range by maximizing a criterion that balances treatment effect estimation and statistical precision. This approach is underpinned by rigorous theoretical justification, addressing both the uncertainty in $L$ selection and its impact on treatment effect estimation. Recognizing the nonregularity of the estimation problem under the null hypothesis, we propose two robust solutions: (a) a convex-hull-based method ({AdaRMST.hulc}) inspired by recent developments \citep{kuchibhotla2024hulc}, and (b) a penalized method ({AdaRMST.ct}) that regularizes estimation under both null and alternative hypotheses, achieving higher power.
(2) For cases where restriction time candidates form a \emph{discrete-time} grid, informed by domain knowledge, we develop the {AdaRMST.dt} estimator. Notably, we prove that pre-specifying a fixed grid yields an adaptive estimator with optimal power, equivalent in large samples to an oracle with prior knowledge of $L^*$. This surprising result underscores the practical advantages of the discrete-time framework.
(3) We bridge the continuous- and discrete-time frameworks via an infill asymptotic regime (``growing grid''), offering practical guidance on grid design and candidate selection for $L$.

We conducted extensive simulation studies across various sample sizes and survival scenarios potentially of practical interest to compare our proposed methods against standard approaches, including canonical RMST analysis, the log-rank test, weighted log-rank tests such as MaxCombo, and recent RMST-based omnibus tests. Results demonstrate that:
{AdaRMST.ct} achieves a strong balance of power (comparable to MaxCombo) and coverage while maintaining well-controlled type I error rates.
{AdaRMST.dt} consistently delivers the best power across scenarios due to the cost-free selection of $L$.

To illustrate its practical utility, we applied our method to a recent oncology clinical trial by \citet{tempero2023adjuvant} comparing the efficacy of adjuvant nab-paclitaxel plus gemcitabine versus gemcitabine alone in 866 patients with surgically resected pancreatic ductal adenocarcinoma. Despite a clear separation of survival curves in the first 18 months (indicating a transient treatment effect), the trial was deemed negative based on the log-rank test, failing to demonstrate sufficient evidence to change the standard of care. Given the difficulty of anticipating treatment effect patterns in advance, adaptive methods like {AdaRMST.ct}, which do not rely on the proportional hazards assumption, offer a powerful and more robust alternative. Reconstructing and reanalyzing the trial data with {AdaRMST.ct} revealed significant clinical implications, showcasing the potential of our approach to uncover meaningful treatment effects that conventional methods might miss.

\subsection{Related Literature}

The challenge of selecting the restriction time $L$ in RMST analysis has been addressed in several recent studies through approaches that evaluate RMST differences at multiple time points or across the entire study duration. For instance, \citet{zhao2016restricted} proposed constructing simultaneous confidence bands for the RMST curve over time, providing a comprehensive view of the treatment effect across different restriction times. \citet{gu2023omnibus} introduced an RMST-based omnibus Wald test that evaluates differences at multiple pre-specified time points, offering a flexible alternative. Other researchers have proposed methods for selecting an optimal time point based on statistical criteria. For example, \citet{horiguchi2018flexible} developed statistical tests that compare RMST at finite pre-specified time points, and then select the time point that maximizes the test statistic. However, this approach lacks rigorous theoretical justification.
Our work builds on these foundations by introducing a set of theoretically justified, data-driven approaches for adaptively selecting $L$ from either a continuous range or a discrete grid of candidate values. These methods ensure robust treatment effect estimation and hypothesis testing while explicitly addressing the uncertainty associated with $L$-selection and the nonregularity of estimation under the null hypothesis.

\section{Method}
\label{sec:method}

This section outlines the methodology underlying our adaptive RMST framework. We begin by introducing notation, setting, and estimands. Then, we introduce the Kaplan-Meier plug-in estimator for the RMST difference and a criterion function for adaptively selecting the optimal restriction time $\Lopt$. Next, we explore the implications of the null hypothesis $\mathsf{H_0}$ highlighting the challenges posed by the non-uniqueness of $\Lopt$ and demonstrating that the treatment effect estimator $\hka^*$ remains robust. Finally, we describe a convex-hull-based approach for constructing asymptotically valid confidence intervals and hypothesis tests.

\subsection{Notation and setting}

The potential outcomes framework is employed in our study \citep{imbens2015causal, hernan2010causal}. We consider two treatment arms, denoted by $A = 0$ (control) or $A = 1$ (treatment), which are randomized at the baseline. The potential outcome under investigation is a time-to-event variable, $T^a$ ($a = 0,1$), representing the time to the event if an individual were assigned to the arm $a$. Let $S^a$ be the survival function of $T^a$.  Because we cannot observe both $T^0$ and $T^1$ for the same individual, one of these potential outcomes is always counterfactual. 

Given a specified restriction time $L$, the Restricted Mean Survival Time (RMST) for arm $a$ is defined as
\[\theta^a(L) = \E[T^a \wedge L],\]
which represents the expected survival time up to $L$ for arm $a$, where $x\wedge y \equiv \min(x,y)$.
The causal estimand of interest is the difference in RMST between the treatment and control arms, defined as
\[\ka(L) = \theta^1(L) - \theta^0(L).\]
In the presence of right-censoring, where the event time is censored by a random time $C^a$, the potential outcome one would then observe under treatment regime $a$ is thus $X^a \equiv T^a \wedge C^a, \Delta^a \equiv \I{T^a \le C^a},$
where $X^a$ is the observed time and $\Delta^a$ is the censoring indicator. 
The observed data for each participant is a triplet $Z \equiv (A, X, \Delta),$ 
where $A$ is the treatment arm, $X$ is the observed time, and $\Delta$ is the censoring indicator. The dataset consists of $n$ i.i.d subjects, with $n_1$ assigned to the treatment arm and $n_0 = n - n_1$ assigned to the control arm. We assume the allocation ratio satisfies $n_1/n \to \beta$ for some $\beta \in (0,1)$ as $n \to \infty$.

The set of candidate restriction times is denoted by the interval 
$\mathcal{L} = [L_{\min}, L_{\max}]$, which may be determined by clinical relevance and regulatory guidelines. When domain experts agree on a reasonable range for $L$, the interval may typically be narrow. However, in cases of uncertainty, $\mathcal{L}$ may be broader. In such scenarios, $L_{\min}$ is chosen as a small but clinically meaningful time, such as $L_{\min} = 3$ months, while $L_{\max}$ is customarily set to exceed the expected follow-up time for most participants, e.g. $L_{\max} = 5$ years. We make the following assumptions throughout the paper:
\begin{assumption}[Setting]
  We assume
  \begin{enumerate}
    \item  Causal Consistency: $(X,\Delta) = (X^A, \Delta^A)$.
    \item  Randomization: $A \indep (T^1, T^0), A \indep (C^1, C^0)$.
    \item Independent Censoring: $T^a\indep C^a\mid A$ for $a =0,1$.
    \item Positivity: $\Prob(T^a \ge L_{\max}) > 0$ and $\Prob(C^a \ge L_{\max}) > 0$ for $a = 0, 1$.
  \end{enumerate}
  \label{assum:setting}
\end{assumption}
Define a population criterion function $\mathbb{M}(L)$ of restriction time $L$:
\[ \mathbb{M}(L) \equiv \ff{\ka^2(L)}{\sigma^2(\hat\ka(L))},\]
where $\hat\ka(L)$ is a chosen estimator of $\ka(L)$ and $\sigma^2(\hat\ka(L))$ (or $\sigma^{\ka,2}_L$ for short) its asymptotic variance. Note that $\sigma^2(\hat\ka(L))$, defined in Equation \eqref{eq:var_ka} below, is a functional of the underlying data generating process and not a function of sample size $n$. Our estimand pair is $(L^*, \ka^*)$, with the primary focus on the treatment effect $\ka^*$, while $L^*$ aids in its interpretation:
\begin{equation}
  \begin{aligned}
    \ka^* \equiv \ka(L^*), \quad L^* &\equiv \argmax_{L \in \mathcal{L}}  \mathbb{M}(L).
  \end{aligned}
  \label{eq:estimand}
\end{equation}
Our estimands have a straightforward and clinically meaningful interpretation: on average, patients in the treatment arm are expected to live longer by $\ka^*$ unit of time (e.g., months or years) compared to those in the control arm, over the period up to the optimal restriction time $L^*$.

Why choose this criterion function? Let $\mathcal{N}(\mu, \sigma^2)$ denote Gaussian distribution and let $\dto$ denote convergence in distribution. Under standard asymptotic theory, we have
\begin{equation}
  \sqrt{n}(\hka_L - \ka_L) \dto \mathcal{N}(0, \sigma^2(\hat\ka(L))).
  \label{eq:var_ka}
\end{equation}
Consequently, the corresponding test statistic is
\[ \ff{\sqrt{n}\hka_L}{\sigma(\hat\ka(L))} \sim \ff{\sqrt{n}\ka_L}{\sigma(\hat\ka(L))} + \mathcal{N}(0, 1).\]
Then, we have the power of the test directly related to the signal-to-noise ratio, also sometimes called ``non-centrality parameter" of the Wald test statistic of no treatment effect, given by:
\[\ff{|\ka_L|}{\sigma(\hat\ka(L))} = \sqrt{\mathbb{M}(L)}.\]
Therefore, maximizing $\mathbb{M}(L)$ is equivalent to maximizing the power of the test under the alternative hypothesis. This provides a theoretically grounded justification for our choice of criterion function.

\subsection{Estimation}
Let $Z_i = (A_i, X_i, \Delta_i)$ be the data triplet for subject $i$. In arm $A = a$, let $T_{a(1)}, T_{a(2)}, \ldots$ represent the ordered distinct event times, with $d_{a(1)}, d_{(a)2}, \ldots$ being the number of events at each event time, allowing for ties. Define the at-risk process $Y^a(t) = \sum_{i = 1}^n \I{A_i = a, X_i \ge t}$, which counts the number of participants at risk just before time $t$ in arm $a$. 
Using these quantities, the Kaplan-Meier estimator \citep{kaplan1958nonparametric} of the survival function 
$S^a(t)$ is%
\[\hat{S}^a(t) \equiv \prod_{j:T_{a(j)} \le t} \left\{1 - \ff{d_{a(j)}}{Y^a(T_{a(j)})}\right\}.\] 
The Kaplan-Meier estimator is a consistent and asymptotically normal estimator of the survival function well-known to attain the semi-parametric efficiency bound for the nonparametric model under standard assumptions \citep[Chapter 4]{andersen1993statistical}, which we adopt throughout this work.

The definition of the optimal restriction time $\Lopt$ depends on the choice of the treatment effect estimator $\hat\ka(L)$. In this work, we use the Kaplan-Meier plug-in estimator:
\begin{equation*}
  \begin{aligned}
    \hat\ka(L) &= \hat\theta^1(L) - \hat\theta^0(L) = \int_0^L \hat S^1(t) \dy{t} -\int_0^L  \hat S^0(t) \dy{t}.
  \end{aligned}
\end{equation*}
Define the sample criterion function as
\[\mathbb{M}_n (L) \equiv \ff{\hat\ka^2(L)}{\hat\sigma^2(\hat\ka(L))}, \]
where the denominator is an estimator of $\sigma^2(\hat\ka(L))$, defined by
\begin{equation}
  \hat\sigma^{\ka,2}_L \equiv \hat\sigma^2(\hat\ka(L)) \equiv n\sum_{a =0}^1 \sum_{j:T_{a(j)} \le L} \ff{d_{a(j)}}{\left[Y^a(T_{a(j)})\right]\left[Y^a(T_{a(j)} )- d_{a(j)}\right]} \left[\hat{\theta}^a(L) - \hat{\theta}^a\left(T_{a\left(j\right)}\right)\right]^2.
  \label{eq:sdkappa}
\end{equation} 
For a given $L$, we set $\mathbb{M}_n (L) = -\infty$ whenever either its numerator or denominator is not estimable from the data, for example, when $L$ is larger than the maximum follow-up time in the sample data.
Then, our estimators for $(L^*, \ka^*)$ are defined as
\begin{equation}
  \begin{aligned}
    \hat\ka^* \equiv \hat\ka(\hat L^*), \quad \hat L^* \equiv \argmax_{L \in {\mathcal{L}}} \mathbb{M}_n (L).
  \end{aligned}
  \label{eq:estimator}
\end{equation}
We now state the asymptotic properties of $\hat L^*$ and $\hat\ka^*$, under the following assumptions.
\begin{assumption}[Smoothness]
  Both $S^1, S^0$ are continuously differentiable. %\footnote{This ensures $\mathbb{M}$ is continuously twice differentiable at $L^*$.}\\
  \label{assum:smooth}
\end{assumption}%
This ensures $S^1, S^0$ are well-behaved, and $\mathbb{M}$ is continuously twice differentiable at $L^*$, which is critical for the derivation of asymptotic properties.
\begin{assumption}[Uniqueness]
  There exists a unique element $L^*$ in the interior of $\mathcal{L}$ such that $\mathbb{M}(L^*) > \sup_{L \notin G} \mathbb{M}(L)$ for every open set $G \subseteq \mathcal{L}$ that contains $ L^*$. 
  \label{assum:unique}
\end{assumption}%
This ensures that $L^*$ is identifiable.
\begin{assumption}[Nonsingularity]
  The second derivative of $\mathbb{M}(L)$ at $L^*$, $\dy^2_L\mathbb{M}(L^*)$, is negative.
  \label{assum:nonsin}
\end{assumption}
This ensures sufficient curvature at $L^*$. 

The following theorem establishes the consistency and asymptotic normality of $\hat L^*$ and $\hat\ka^*$.
\begin{proposition}[Consistency and Asymptotic Normality]
  Under Assumptions \ref{assum:setting}, \ref{assum:smooth}, \ref{assum:unique}, and \ref{assum:nonsin}, the estimators $\hat L^*$ and $\hat\ka^*$ are $\sqrt{n}$-consistent and asymptotically normal. Specifically, as $n \to \infty$,
  \[\sqrt{n}(\hL^* - L^*) \dto \mathcal{N}(0, \sigma^{L^*, 2}), \sqrt{n}(\hka^* - \ka^*) \dto \mathcal{N}(0, \sigma^{\ka^*, 2}),\]
  where $\sigma^{L^*, 2}, \sigma^{\ka^*, 2}$ are the respective asymptotic variances.
  \label{thm:can}
\end{proposition}
The proof is provided in the Supplementary Material, where details of the asymptotic variances emerge in the derivation. 

\subsection{Null Hypothesis}

We consider the null hypothesis and the alternative hypothesis as follows:
\[\mathsf{H_0}: S^1_t = S^0_t \text{ for all } t \quad \text{v.s.} \quad \mathsf{H_1}: S^1_t \ne S^0_t \text{ for some } t.\]
Note that under $\mathsf{H_0}$, the survival functions for the treatment and control arms are identical across the interval $\mathcal{L}$. This implies for all $L$,
\[\ka(L^*) = \ka(L) = 0 \text{ and } \mathbb{M}_{L^*} = \mathbb{M}_L = 0.\] 
As a result, the optimal restriction time $\Lopt$ is not uniquely defined under $\mathsf{H_0}$, which violates Assumption~\ref{assum:unique}. Consequently, the asymptotic properties derived in Proposition~\ref{thm:can} no longer hold, as they rely on the uniqueness of $\Lopt$. 
To address this limitation, we investigate the asymptotic behavior of $\hka^*$ under $\mathsf{H_0}$. Despite the challenges posed by the lack of a unique $\Lopt$, we show that $\hka^*$ retains desirable statistical properties under the null hypothesis. Let $O_p(1)$ denote a sequence that is bounded in probability, and $o_p(1)$ denote a sequence that converges to $0$ in probability. The results are summarized in the following propositions.
\begin{proposition}[$\sqrt{n}$-Consistency of $\hka^*$ under $\mathsf{H_0}$]
  Under Assumptions \ref{assum:setting} and \ref{assum:smooth}, $\hka^*$ remains $\sqrt{n}$-consistent for $\ka^* = 0$ under the null hypothesis $\mathsf{H_0}$, i.e., $\sqrt{n}(\hka^* - \ka^*) = O_p(1).$%
  \label{prop:consistency}
\end{proposition}%
This result ensures that $\hka^*$ converges in probability to the true value $\ka^* = 0$ at the standard parametric rate, even when $\mathsf{H_0}$ holds.
\begin{proposition}[Asymptotic Median-Unbiasedness under $\mathsf{H_0}$]
  Suppose Assumptions \ref{assum:setting} and \ref{assum:smooth} hold. Under the null hypothesis $\mathsf{H_0}$, the estimator $\hka^*$ is asymptotically median-unbiased. Specifically,
  \begin{equation*}
    \lim_{n \to \infty} \Pr\left[ \max\left\{0, \ff{1}{2} - \min\{\Pr(\hka^* \ge \ka^*), \Pr(\hka^* \le \ka^*)\} \right\}\right] = 0.
    \label{eq:amu}
  \end{equation*}
  \label{prop:median}
\end{proposition}
This property indicates that the median of the sampling distribution of $\hka^*$ aligns asymptotically with the estimand $\ka^* = 0$, which is a desirable property for constructing convex-hull based confidence intervals under $\mathsf{H_0}$. Together, these properties support the reliability of 
$\hka^*$ as an estimator under the null hypothesis, despite the theoretical complications arising from the non-uniqueness of $\Lopt$.

\subsection{Confidence Interval and Hypothesis Testing}
We describe below a unified algorithm to construct an asymptotically valid $(1-\alpha)$-confidence interval for the true treatment effect $\ka^*$ under both the null and alternative hypotheses, despite their differing asymptotic behaviors. This confidence interval, $\CIh_{n,\alpha}$, is based on the \emph{HulC} (Hull-based Confidence interval) approach introduced by \citet{kuchibhotla2024hulc}, which requires only knowledge of the asymptotic median bias of the estimators. For the hypothesis testing, the null hypothesis $\mathsf{H_0}$ is rejected if the constructed confidence interval does not contain zero. 

In addition, we propose a slightly anti-conservative variant, $\tCIh_{n,\alpha}$, which exhibits improved power and practical coverage performance. Let $\floor{x}$ denote the floor function of a real number $x$, which is the largest integer less than or equal to $x$. Let $\ceil{x}$ denote the ceiling function, which is the smallest integer greater than or equal to $x$. The procedure for constructing these confidence intervals and performing the dual hypothesis test is summarized in Algorithm~\ref{alg:CIh}. 
\begin{algorithm}[ht]
    \caption{Convex-Hull Based Confidence Intervals and the Dual Test for \(\kappa^*\)}
    \label{alg:CIh}
    
    \KwIn{Significance level \(\alpha \in (0,1)\), sample data \(\{Z_i = (A_i, X_i, \Delta_i)\}_{i=1}^n\), and confidence interval type (\(\CIh_{n, \alpha}\) or \(\tCIh_{n, \alpha}\)).}
    \KwOut{The confidence interval, \(\CIh_{n, \alpha}\) or \(\tCIh_{n, \alpha}\), and the dual hypothesis test result.}
    \textbf{Step 1:} Compute the integer \(B\) as follows:\;
    \Indp
    - \( B = \ceil{1 - \ln(\alpha) / \ln(2)}\) for \(\CIh_{n, \alpha}\).\;
    - \( B = \floor{1 - \ln(\alpha) / \ln(2)}\) for \(\tCIh_{n, \alpha}\).\;
    \Indm
    
    \textbf{Step 2:} Randomly partition the sample data into \( B \) folds.  
    For each \(j\) (\(1 \le j \le B\)), compute the fold-specific estimator \(\hat{\kappa}^*_j\) using only the \(j\)-th fold data (see Equation~\eqref{eq:estimator}).\;
    
    \textbf{Step 3:} Construct the confidence interval as
    $\bigl[\min_{1 \leq j \leq B} \hat{\kappa}^*_j,\;\max_{1 \leq j \leq B} \hat{\kappa}^*_j\bigr].$\;
    
    \textbf{Step 4:} Conduct the hypothesis test. Reject \(\mathsf{H_0}\) if the above interval does not contain zero.\;
    
\end{algorithm}

The theoretical properties of $\CIh_{n, \alpha}$ and $\tCIh_{n, \alpha}$ under both $\mathsf{H_0}$ and $\mathsf{H_1}$
are formalized in the following proposition.
\begin{proposition}[Asymptotic properties of $\CIh_{n,\alpha}$ and $\tCIh_{n,\alpha}$ under both $\mathsf{H_0}$ and $\mathsf{H_1}$]
  Under Assumptions \ref{assum:setting}, \ref{assum:smooth} for both $\mathsf{H_0}$ and $\mathsf{H_1}$, and Assumptions \ref{assum:unique} and \ref{assum:nonsin} for $\mathsf{H_1}$, the confidence intervals constructed using Algorithm~\ref{alg:CIh} have the following properties:
  \begin{enumerate}
    \item $\lim_{n\to \infty} \Pr\left(\ka^* \in \CIh_{n,\alpha} \right) \ge 1 - \alpha$.
    \item $\lim_{n\to \infty} \Pr\left(\ka^* \in \tCIh_{n,\alpha}\right) \ge 1 - 2\alpha$.
  \end{enumerate}
  \label{prop:hulc}
\end{proposition}
While $\CIh_{n,\alpha}$ guarantees asymptotic validity, we recommend using 
$\tCIh_{n,\alpha}$ in practice due to its improved power and practical coverage properties. Notably, the coverage probability lower bound of $1-2\alpha$ is not sharp, as the HulC method is generally conservative in finite samples. For instance, when $\alpha = 0.05$, $\tCIh_{n,0.05}$ with $B = 5$ provides a theoretical coverage of at least $1 - 2^{1-B} = 93.75\%$, which exceeds $1 - 2\times 0.05 = 90\%$. The practical coverage is further closer to 95\% due to its conservativeness.

To summarize, point estimate $\hka^*$ of $\ka^*$ is obtained by Equation \eqref{eq:estimator}, while its HulC confidence interval is obtained by Algorithm~\ref{alg:CIh}. The corresponding selected restriction time $\hL^*$ is the maximizer in Equation \eqref{eq:estimator}, assisting in interpreting $\hka^*$. 

\section{Penalization}
\label{sec:pen}

In this section, we introduce a penalization-based method for enhancing the criterion function, offering both theoretical and practical advantages. This approach modifies the estimand slightly while preserving its interpretability, and ensures that the corresponding treatment effect estimator maintains asymptotic normality under both null and alternative hypotheses.

We define the penalized criterion function adding a concave quadratic penalty term:
\[ \mathbb{M}^{\dagger}(L) \equiv \mathbb{M}(L) - c\left(L - \tilde{L}\right)^2= \ff{\ka^2_L}{\sigma^{\ka,2}_L} - c\left(L - \tilde{L}\right)^2,\]
where $c > 0$ is a small, pre-specified constant, and $\tilde{L} \in ( L_{\min}, L_{\max})$ is an initial guess for the optimal restriction time.

The modified estimands $(L^\dagger, \ka^\dagger )$ are defined as 
\begin{equation}
  \begin{aligned}
    \ka^\dagger \equiv \ka(L^\dagger),\quad L^\dagger \equiv \argmax_{L \in \mathcal{L}} \mathbb{M}^{\dagger}(L).
  \end{aligned}
  \label{eq:M_pen}
\end{equation}
When there is no penalization ($c = 0$), $(L^\dagger, \ka^\dagger )$ reduces to $(L^*, \ka^* )$. However, when $c>0$,  the penalty term favors values of $L$ closer to $\tilde{L}$.
Despite the modification, the interpretation remains intuitive: on average, patients in the treatment arm are expected to live $\ka^\dagger$ units of time (e.g., months or years) longer than those in the control arm over the period up to $L^\dagger$.

\begin{remark}[Benefits of Penalization]
\textit{Uniqueness Under $\mathsf{H_0}$}: A key theoretical advantage of penalization is that under the null hypothesis $\mathsf{H_0}$, the penalized criterion function $\mathbb{M}^\dagger$ has a unique maximizer, $L^\dagger = \tilde{L}$ (see Figure~\ref{fig:scenario_null} in the Supplementary Material). This resolves the non-uniqueness issue present in the unpenalized method and allows for a unified and more efficient inferential procedure under both null and alternative hypotheses. 
\end{remark}

The penalized estimators are defined as
\begin{equation}
  \begin{aligned}
    \hat\ka^\dagger \equiv \hat\ka(\hat L^\dagger), \quad \hat L^\dagger \equiv \argmax_{L \in {\mathcal{L}}} \mathbb{M}_n (L) - c\left(L - \tilde{L}\right)^2.
  \end{aligned}
  \label{eq:estimator_pen}
\end{equation}
Under the following assumptions, the penalized estimators retain desirable asymptotic properties:
\begin{assumbis}{assum:unique}[Uniqueness]
  There exists a unique element $L^\dagger$ in the interior of $\mathcal{L}$ such that $\mathbb{M}^\dagger(L^\dagger) > \sup_{L \notin G} \mathbb{M}^\dagger(L)$ for every open set $G \subseteq \mathcal{L}$ that contains $ L^\dagger$. 
  \label{assum:unique_pen}
\end{assumbis}

\begin{assumbis}{assum:nonsin}
  The second derivative of $\mathbb{M}^\dagger(L)$ at $L^\dagger$, $\dy^2_L\mathbb{M}^\dagger(L^\dagger)$, is negative.
  \label{assum:nonsin_pen}
\end{assumbis}

\begin{theorem}[Asymptotic Properties Under Penalization]
  Under Assumptions \ref{assum:setting}, \ref{assum:smooth}, \ref{assum:unique_pen}, and \ref{assum:nonsin_pen}, the penalized estimators have the following properties:
  \begin{enumerate}
    \item  Under $\mathsf{H}_1$, $\hka^\dagger$ and $\hL^\dagger$ are $\sqrt{n}$-consistent and asymptotically normal estimators of $\ka^\dagger, L^\dagger$, respectively.
    \item Under $\mathsf{H}_0$, $|\hL^\dagger -  L^\dagger| = O_p(n^{-1})$, and $\sqrt{n}(\hka^\dagger - \ka^\dagger) = \sqrt{n}(\hka(\tilde{L}) - \ka(\tilde{L})) + o_p(1) \dto \mathcal{N}(0, \sigma^{\ka,2}_{\tilde{L}})$.
  \end{enumerate}
  \label{thm:asymp_pen}
\end{theorem}

\textbf{Bootstrap Confidence Intervals and Testing} \quad To construct confidence intervals for $\dL$ and $\dka$, as well as perform hypothesis testing, we employ a bootstrap procedure formalized in Algorithm~\ref{alg:CIh_pen}. Bootstrap methods often outperform Wald-type intervals in finite samples, particularly in scenarios where asymptotic approximations may struggle due to small sample sizes. Its validity is proved in Proposition~\ref{prop:bootstrap}. 

\begin{algorithm}[ht]
    \caption{Bootstrap Confidence Intervals for $\kappa^\dagger, L^\dagger$ and the Dual Test for $\kappa^\dagger$}
    \label{alg:CIh_pen}
    
    \KwIn{Significance level $\alpha \in (0,1)$, sample data $O_0 = \{ Z_i = (A_i, X_i, \Delta_i) \}_{i=1}^n$, number of bootstrap resamples $B$.}
    \KwOut{The confidence intervals $\CIb_{\kappa^\dagger, n, \alpha}, \CIb_{L^\dagger, n, \alpha}$, and the hypothesis test result.}
    
    \textbf{Step 1:} Generate bootstrap samples and estimates:\\
    \For{$b \gets 1$ \KwTo $B$}{
      Draw a bootstrap sample $O^*_b = \{ Z_i^*\}_{i=1}^n$ by sampling with replacement from \( O_0 \).\\
      Calculate the estimates $(\hat{\kappa}^\dagger_b, \hat{L}^\dagger_b)$ based on Equation \eqref{eq:estimator_pen} with the bootstrap sample \( O^*_b \).
    }
    
    \textbf{Step 2:} Construct the percentile bootstrap confidence intervals:
    \[
      \CIb_{\kappa^\dagger, n, \alpha} 
      = 
      \bigl[\, q_{\alpha/2}(\{\hat{\kappa}^\dagger_b\}_{b = 1}^B), \; q_{1-\alpha/2}(\{\hat{\kappa}^\dagger_b\}_{b = 1}^B)\bigr], \quad \CIb_{L^\dagger, n, \alpha} 
      = 
      \bigl[\, q_{\alpha/2}(\{\hat{L}^\dagger_b\}_{b = 1}^B), \; q_{1-\alpha/2}(\{\hat{L}^\dagger_b\}_{b = 1}^B)\bigr],
    \]
    where $q_{\alpha}(\cdot)$ is the $\alpha$-th quantile of a list of statistics.
    
    \textbf{Step 3:} Conduct the hypothesis testing:\\
    Reject the null hypothesis \(\mathsf{H_0}\) if 
    \(
      \CIb_{\kappa^\dagger, n, \alpha} 
    \)
    does not contain zero.
    
\end{algorithm}

\begin{proposition}[Bootstrap Consistency]
  Under Assumptions \ref{assum:setting}, \ref{assum:smooth}, \ref{assum:unique_pen}, and \ref{assum:nonsin_pen}, the bootstrap confidence intervals $\CIb_{\ka^\dagger, n, \alpha}, \CIb_{L^\dagger, n, \alpha}$ by Algorithm \ref{alg:CIh_pen} have the following properties:
  \begin{enumerate}
    \item $\lim_{n\to \infty} \Pr\left(\ka^\dagger \in \CIb_{\ka^\dagger, n, \alpha}\right) \ge 1 - \alpha$ under both $\mathsf{H}_0, \mathsf{H}_1$.
    \item $\lim_{n\to \infty} \Pr\left(L^\dagger \in \CIb_{L^\dagger, n, \alpha} \right) \ge 1 - \alpha$ under $\mathsf{H}_1$.
  \end{enumerate}
  \label{prop:bootstrap}
\end{proposition}

The penalization approach proposed in this section improves the estimation efficiency for treatment effect and guarantees its asymptotic normality under both the null and alternative hypotheses. It further enables robust bootstrap-based inferential procedures without requiring prior knowledge of which hypothesis holds. Furthermore, it facilitates the estimation of the optimal restriction time under the null hypothesis at a fast convergence rate. As a result, hypothesis testing becomes $L$-selection-variability-free under $\mathsf{H}_0$, mirroring the scenario in which the restriction time is fixed at $\tilde{L}$ from the outset. This stands in contrast to the difficulties encountered when non-uniqueness arises in the absence of penalization.

In practice, the choice of $\tilde{L}$ should be guided by domain knowledge; if no prior information is available, the midpoint of the candidate interval $(L_{\min} + L_{\max})/2$ serves as a reasonable default. The penalty parameter $c$, which controls the strength of the penalization, balances the reliance on prior knowledge versus the adaptivity to data. For instance, smaller values of $c$ allow the data to drive the selection of $L$, whereas larger values enforce a closer alignment with $\tilde{L}$ and reduce adaptability. Based on simulation studies, we recommend setting $c$ to
\[c = 0.002\times\ff{16}{(L_{\max} - L_{\min})^2}\times \left(\ff{1 \text{ unit}}{1 \text{ year}}\right)^2,\]
which scales appropriately with the length of $L$ and the time units used in the analysis.
Figures \ref{fig:pen_300}-\ref{fig:pen_1000} in the Supplementary Material further illustrate how the statistical power changes with the penalization parameter $c$ across scenarios in simulation.

\section{Discrete-time Grid of Candidate Restriction Times}
\label{sec:dt}

In this section, we introduce a method for selecting the optimal restriction time $L^*$ from a pre-specified discrete-time grid of candidate restriction times. This approach is particularly advantageous when researchers have prior knowledge of a few plausible candidate restriction times. Notably, the fixed grid provides valuable information and significantly reduces the search space, creating a ``nearly-free lunch" in terms of selection cost. Consequently, we can construct a confidence interval for $\ka^*$ at the selected $L$. In contrast, \citet{horiguchi2018flexible} focus on discrete-time grids but rely on simultaneous confidence bands, which may be more conservative and less powerful.

We define the discrete-time grid as $\mathcal{L} = \{L_1, L_2, \ldots, L_m\}$, where $L_1 = L_{\min}$ and $L_m = L_{\max}$. To address the issue of non-uniqueness under the null hypothesis, we adopt a penalized approach similar to the method described in the previous section. The corresponding criterion function is given by
\[\mathbb{M}_j^{\dagger} \equiv \mathbb{M}(L_j) - c (L - \tilde{L})^2, j = 1, 2, \ldots, m,\] 
where $\tilde{L} = L_{\tilde{j}}$ represents a pre-specified initial choice of restriction time. By default, we set $\tilde{j} = \floor{(m+1)/2}$. Without prior knowledge, we suggest setting $c = 0.005\times\ff{16}{(L_{\max} - L_{\min})^2}\times \left(\ff{1 \text{ unit}}{1 \text{ year}}\right)^2$ as the default value according to our simulation.
The optimal restriction time is then defined as $L^\dagger \equiv L_{j^\dagger}$, where $j^\dagger = \argmax_{j\in \mathcal{L}} \mathbb{M}^\dagger_j$, and $\kappa^\dagger \equiv \ka(L^\dagger)$. The estimators for $(L^\dagger, \ka^\dagger)$ are given by
\begin{equation}
  \begin{aligned}
    \hat\ka^\dagger &\equiv \hat\ka(\hat L^\dagger), \quad \hat L^\dagger \equiv L_{\hat j^\dagger},\\
    \hat j^\dagger &= \argmax_{j\in \{1,\ldots,m\}} \mathbb{M}_n^\dagger(L_j) \equiv \argmax_{j\in j\in \{1,\ldots,m\}} \mathbb{M}_n(L_j) -  c (L_j - \tilde{L})^2. 
  \end{aligned}
  \label{eq:estimator_grid}
\end{equation}
We make the following assumptions:
\begin{assumbis}{assum:unique_pen}
  There exists a unique element $L_{j^\dagger}$ in $\mathcal{L}$ such that $\mathbb{M}^\dagger(L_{j^\dagger}) > \mathbb{M}^\dagger(L_j)$ for all $j \ne j^\dagger$. 
  \label{assum:unique_dt}
\end{assumbis}
This guarantees the existence of a unique solution, preventing ambiguities in the selection process.

The asymptotic properties of the estimators $\hat L^\dagger$ and $\hat\ka^\dagger$  under both the null hypothesis  ($\mathsf{H}_0: S^1_{L_j} = S^0_{L_j}$ for all $j = 1, \ldots, m$) and the alternative hypothesis are established in the following theorem.
\begin{theorem}[Discrete-time Estimator]
  Under Assumptions \ref{assum:setting}, \ref{assum:smooth}, and \ref{assum:unique_dt}, the estimators defined in Equation \eqref{eq:estimator_grid} have the following properties:
  \begin{enumerate}
    \item $\hL^\dagger$ is a consistent estimator of $L^\dagger$. In particular, $n^{\upsilon }|\hL^\dagger - L^\dagger| = o_p(1)$ for arbitrary $\upsilon \ge 0$.
    \item $\hka^\dagger$ is a $\sqrt{n}$-consistent and asymptotically normal estimator of $\ka^\dagger$, with 
  \[\sqrt{n}(\hka^\dagger - \ka^\dagger) = \sqrt{n}(\hka_{L^\dagger} - \ka_{L^\dagger}) + o_p(1) \dto \mathcal{N}(0, \sigma_{L^\dagger}^{\ka, 2}),\]
  and $\hat\sigma_{\hat{L}^\dagger}^{\ka, 2}$ is a consistent estimator for $\sigma_{L^\dagger}^{\ka, 2}$.
  \end{enumerate}
  \label{thm:dt}
\end{theorem}
The rate of convergence in part (1), $o_p(n^{-1/2})$, is faster than the typical parametric rate of $O_p(n^{-1/2})$, as observed in the continuous-time setting. This faster rate underscores the efficiency of the selection process on the grid and directly leads to the no-selection-cost property in terms of the asymptotic variance of $\hka^\dagger$.
Using Theorem~\ref{thm:dt}, we can construct an asymptotically valid $(1-\alpha)$-confidence interval for $\ka^\dagger$: \[[\hka^\dagger - z_{1-\alpha/2}\hat\sigma_{\hat{L}^\dagger}^{\ka}/\sqrt{n}, \hka^\dagger + z_{1-\alpha/2}\hat\sigma_{\hat{L}^\dagger}^{\ka}/\sqrt{n}],\] where $z_{1-\alpha/2}$ is the $1-\alpha/2$ quantile of the standard normal distribution. Hypothesis testing can then be performed by checking whether the confidence interval contains zero. 
The point estimate of the optimal restriction time, $\hat{L}^\dagger$, aids in the interpretation of $\hat{\kappa}^\dagger$. Due to the super-efficiency of $\hat{L}^\dagger$, its $\sqrt{n}$-asymptotic variance is zero, i.e. $\sqrt{n}(\hat{L}^\dagger - {L}^\dagger) \pto 0$, leading to a confidence interval that is a singleton: $\{\hat{L}^\dagger\}$.

\subsection{Balancing Grid Density: Insights and Guidelines}

The primary advantage of using a discrete-time grid is that it avoids any asymptotic cost for selection. However, there are potential drawbacks to consider. A grid that is too sparse might exclude restriction times with a much larger value of $\mathbb{M}$. Conversely, an overly dense grid given a small sample size could create challenges for correctly selecting the true optimal restriction time. Consequently, choosing an appropriate grid density is beneficial, and researchers need practical guidelines for this task.

To address this, we propose a rule-of-thumb for selecting the grid by investigating estimator properties under an \emph{infill} asymptotic regime. Specifically, we consider the case where the number of grid points, $m_n$, increases with the sample size $n$. We show that as long as $m_n$ grows at a controlled rate, the discrete-time grid remains cost-free for selection, preserving the method's efficiency.

We formalize these insights in the following proposition. To distinguish between the discrete-time and continuous-time settings, let $L^\dagger_{\ct}$ denote the optimal restriction time in the continuous-time setting (as defined in Equation \eqref{eq:M_pen}), and let $\ka^\dagger_{\ct} \equiv \ka(L^\dagger_{\ct})$. Note that the estimands for the discrete settings change as the grid becomes denser.

\begin{proposition}[Asymptotic Results Under Infill Regime]
  Suppose the grid has $m_n$ points, with $m_n$ increasing with the sample size $n$, and assume there exist constants $\bar{k} > 0$ and $\underline{k} > 0$ such that for all $j = 1, \ldots, m_n-1$, $\underline{k}/(m_n-1) \le L_{j+1} - L_j \le \bar{k}/(m_n-1).$ Let $m_n = kn^\gamma$, where $k,\gamma$ are positive numbers. Suppose Assumptions \ref{assum:setting}, \ref{assum:smooth}, and \ref{assum:unique_pen} hold, and Assumption \ref{assum:unique_dt} holds for each $n$. Then, as $n \to \infty$,
  \begin{enumerate}
    \item (Estimand Consistency) $L^\dagger \to L^\dagger_{\ct}, \ka^\dagger \to \ka^\dagger_{\ct}$.
    \item  (Fast Rate for $L^\dagger$) If $\gamma < 1/4$ (i.e., $m_n = o_p(n^{1/4})$), then $n^\upsilon|\hat{L}^\dagger - L^\dagger| = o_p(1)$ for arbitrary $\upsilon \ge 0$, and 
    \[\sqrt{n}(\hka^\dagger - \ka^\dagger) = \sqrt{n}(\hka_{L^\dagger} - \ka_{L^\dagger}) + o_p(1) \dto \mathcal{N}(0, \sigma_{L^\dagger}^{\ka, 2}).\]
    \item (Slow Rate for $L^\dagger_{\ct}$) When $\gamma < 1/4$, $|\hat{L}^\dagger - L^\dagger_{\ct}| = O_p(n^{-\gamma}) \pto 0$.
  \end{enumerate}
  \label{prop:grid}
\end{proposition}
Part (1) of Proposition~\ref{prop:grid} demonstrates that as the grid becomes denser, the estimands in the discrete-time and continuous-time settings converge.
Part (2) establishes the importance of controlling the growth rate of $m_n$. A rule-of-thumb emerges: \emph{Selecting a grid with approximately $n^{1/4}$ to $2n^{1/4}$ restriction time points.} The upper bound $1/4$ for $\gamma$ might be improvable. This remains an avenue for future work.
Part (3) indicates that with a controlled growth rate for $m_n$, the rate of convergence to the continuous-time estimand is slower. This tradeoff highlights a subtle difference between the discrete and continuous settings: while the discrete-time grid ensures a fast rate for $L^\dagger$, there is a rate-penalty for the estimation of $L^\dagger_{\ct}$, which will generally be slower than the parametric rate observed in the continuous-time settings.

\section{Simulation}
\label{sec:simulation}

We conducted an extensive simulation study to evaluate the performance of the proposed methods in comparison to existing methods from the literature. The evaluation focuses on type I error, power, and coverage probability, with the aim of providing recommendations for practical application. The methods and their configurations are detailed below.

\textbf{Proposed Methods:}\\
(1) \texttt{AdaRMST.ct}: Penalization method with the continuous-time restriction candidate set $\mathcal{L} = [0.2, 4.2]$ years and bootstrap confidence intervals $\CIb_{\ka^\dagger,n,\alpha}$. We apply the default setting with $c = 0.002$, $\tilde{L}$ equal to midpoint $2.2$, and $B = 1000$ bootstrap resamples.\\ 
(2) \texttt{AdaRMST.dt}: Penalization method with a discrete-time grid $\mathcal{L}$ consisting of 10 equally spaced points between $0.2$ and $4.2$. Wald confidence intervals are used with default settings $c = 0.005$ and $\tilde{L} = 1.98$, the midpoint of the grid.\\ 
(3) \texttt{hulc}: We use the anti-conservative version of convex-hull based confidence interval, $\tCIh_{n,\alpha}$.

\textbf{Comparison Methods:}\\
(1) \texttt{rmst}: Traditional RMST estimator with the restriction time set to the maximum follow-up time, as implemented in the \texttt{R} package \texttt{survRM2} \citep{r2021, survRM2}.\\ 
(2) \texttt{logrank}: Log-rank test \citep{schoenfeld1981asymptotic}, a standard and widely used method for comparing survival curves.\\ 
(3) \texttt{maxcombo}: A popular omnibus test that combines multiple Fleming-Harrington (FH) weighted log-rank tests with different weight functions, including FH(0,0), FH(0,1), FH(1,0), FH(1,1) \citep{nph}. This approach is tailored to detect treatment effects under a variety of hazard ratio patterns (e.g. early, middle, delayed), making it a valuable comparator for the proposed methods that aim to accommodate non-proportional hazards adaptively.\\ 
(4) \texttt{hctu}: Omnibus RMST test based on multiple restriction times using 10 equally spaced points \citep{survRM2adapt,horiguchi2018flexible}. This recent method addresses the sensitivity of RMST estimates to the choice of restriction time by considering a range of values, aligning closely with the objectives of our proposed methods. We use the \texttt{R} package \texttt{survRM2adapt} developed by the authors.\\
(5) \texttt{oracle}: The Oracle RMST estimator uses the true optimal restriction time $L^*$, as defined in Equation \eqref{eq:estimand}, which is determined directly from the true population criterion function rather than being estimated from the data. It serves as an ideal benchmark with exactly zero selection cost and maximum power. Although not accessible in practical settings, it provides a theoretical upper bound for performance and highlights the potential benefits of accurately identifying the optimal restriction time.

\textbf{Simulation Design} \quad We examined the performance under sample sizes of $n = 300$, $600$, and $1000$, with an allocation ratio of 1:1. For two-sided tests, a significance level of $\alpha = 0.05$ was used, corresponding to a confidence level of $1-\alpha = 0.95$ for confidence intervals for the treatment effect $\kappa^*$ or $\kappa^\dagger$.
The censoring time distribution was exponential with a rate of $1/2$ in both arms, yielding approximately 33\% censoring across scenarios. The control arm's event times followed an exponential distribution with a rate of $1$, consistent across all scenarios. Administrative censoring was set at 5 years for both arms. 

For the treatment arm, we considered nine scenarios designed to capture common patterns observed in randomized trials: (1) \texttt{null}: No treatment effect (null hypothesis). (2) \texttt{ph}: Proportional hazards. (3) \texttt{early}: Early treatment effects, where the benefit diminishes over time. (4) \texttt{tran}: Transient treatment effects, where the benefit appears temporarily and then disappears. (5) \texttt{cs}: Crossing survival curves. (6) \texttt{msep}: Middle separation of survival curves. (7) \texttt{delay\_1}: Moderately delayed treatment effects, with a gradual onset of benefit. (8) \texttt{delay\_2}: Heavily delayed treatment effects, where the benefit takes a longer time to emerge. (9) \texttt{delaycon}: Delayed treatment effects with converging survival curves, modeling a scenario where the benefit emerges late and then diminishes. These scenarios were generated using piecewise exponential distributions to allow flexible hazard patterns. Their corresponding survival curves are illustrated in Figure~\ref{fig:surv}. Detailed parameters for each scenario, and visualizations of hazard functions, criterion functions, optimal restriction times, and treatment effect function $\ka_t$ are provided in the Supplementary Material.
\begin{figure}
  \centering
  \includegraphics[width=0.8\linewidth]{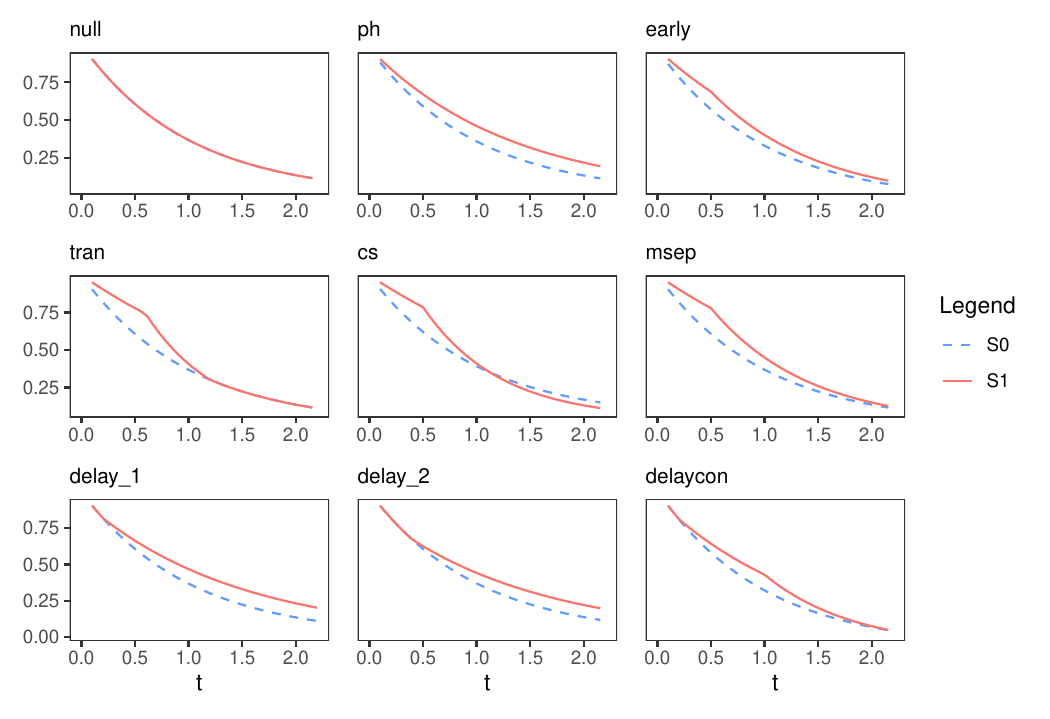}
  \caption{Survival curves for the nine scenarios. S0: control arm, S1: treatment arm. \texttt{null}: No treatment effect (null hypothesis). \texttt{ph}: Proportional hazards. \texttt{early}: Early treatment effects. \texttt{tran}: Transient treatment effects. \texttt{cs}: Crossing survival curves. \texttt{msep}: Middle separation. \texttt{delay\_1}: Moderately delayed treatment effects. \texttt{delay\_2}: Heavily delayed treatment effects. \texttt{delaycon}: Delayed treatment effects with converging survival curves. }
  \label{fig:surv}
\end{figure}

For each scenario and sample size, we generated $N = 2000$ datasets. The empirical type I error, power, and coverage probability (for RMST-based methods) were calculated for each method using an \texttt{R} implementation. 

\textbf{Simulation Results} 
\quad A subset of the results is visualized in Figure~\ref{fig:simu_res},  while additional findings are provided in Supplementary Material. The computation time for \texttt{AdaRMST} is approximately 2 minutes per dataset on a single CPU core of a 2023 MacBook Pro, with substantial reductions achievable through paralleling the bootstrap procedure.

The left panel of Figure~\ref{fig:simu_res} shows the \emph{power} of the proposed methods and comparison methods under different scenarios and sample sizes. The type I error is shown in the \texttt{null} scenario. \texttt{AdaRMST.ct} demonstrates power comparable to \texttt{maxcombo} and performs robustly across scenarios. \texttt{AdaRMST.dt}, however, achieves the highest power overall, closely approaching the performance of the \texttt{oracle} method. This finding aligns with Theorem~\ref{thm:dt}, which suggests that a fixed grid incurs minimal selection cost. In contrast, \texttt{hulc} is generally conservative compared to \texttt{AdaRMST}, though it outperforms \texttt{logrank} and \texttt{rmst} in settings with early treatment effects or crossing survival functions. \texttt{hctu} exhibits strong power but is generally less powerful than \texttt{AdaRMST.dt} across most scenarios. Notably, all methods demonstrate good type I error control, though \texttt{AdaRMST.dt} shows slight inflation when $n = 300$.
\begin{figure}
  \centering
  \includegraphics[width=1\linewidth]{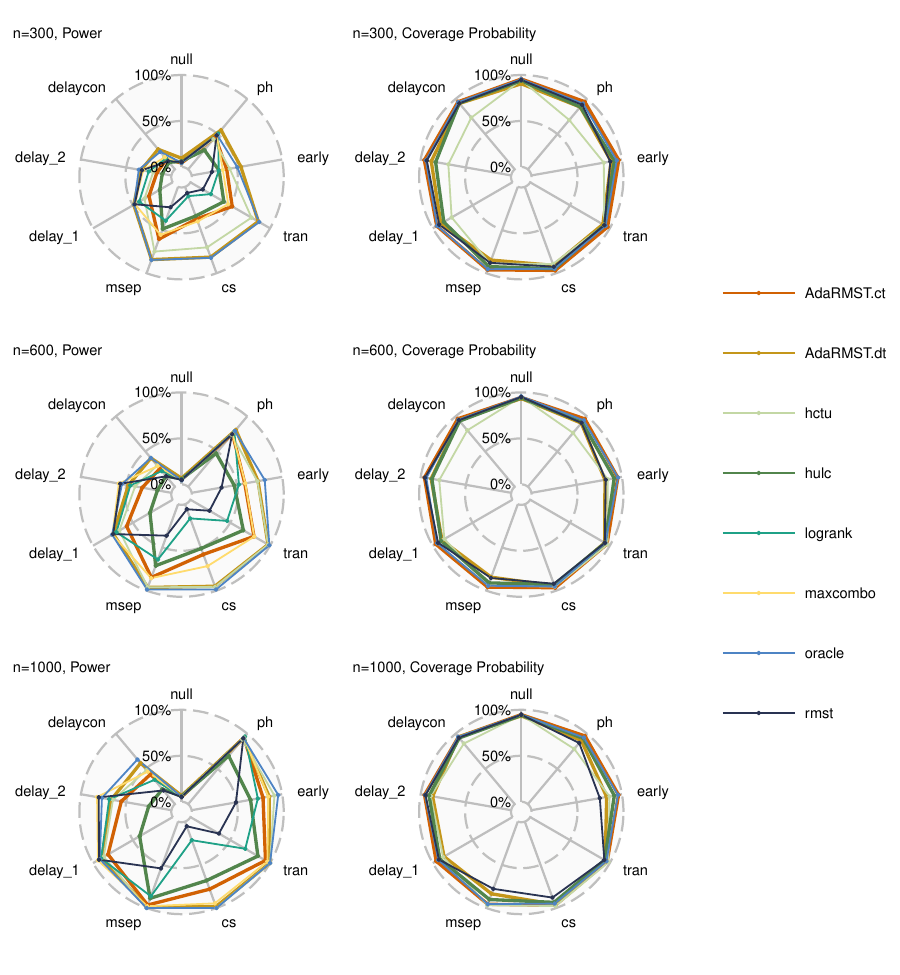}
  \caption{Power and Coverage Probability of the proposed methods and comparison methods under different scenarios and sample sizes. The significance level is $\alpha = 0.05$, corresponding to a confidence level of $1-\alpha = 95\%$ for confidence intervals for the treatment effect $\kappa^*$ or $\kappa^\dagger$. Type I error of each method is shown in the Power plots in the \texttt{null} scenario axis. (\texttt{null}: No treatment effect (null hypothesis). \texttt{ph}: Proportional hazards. \texttt{early}: Early treatment effects. \texttt{tran}: Transient treatment effects. \texttt{cs}: Crossing survival curves. \texttt{msep}: Middle separation. \texttt{delay\_1}: Moderately delayed treatment effects. \texttt{delay\_2}: Heavily delayed treatment effects. \texttt{delaycon}: Delayed treatment effects with converging survival curves. )}
  \label{fig:simu_res}
\end{figure}

The right panel of Figure~\ref{fig:simu_res} illustrates the \emph{coverage probability} of the confidence intervals for the treatment effect estimands, $\kappa^*$ or $\kappa^\dagger$, at the nominal $95\%$ level. Among the methods, \texttt{AdaRMST.ct} achieves the most accurate coverage, closely approximating the nominal $95\%$ level. In contrast, \texttt{maxcombo}, despite its strong statistical power, lacks a corresponding coherent effect size estimate and is therefore excluded from this evaluation. Both \texttt{AdaRMST.dt} and \texttt{hulc} exhibit moderate under-coverage, whereas \texttt{hctu} shows substantial under-coverage, particularly in scenarios with small sample sizes.

Tables \ref{tab:L_n300}, \ref{tab:L_n600}, and \ref{tab:L_n1000} in the Supplementary Material summarize the finite-sample performance of the proposed optimal restriction time estimators, $\hL^*$ and $\hdL$, including \texttt{AdaRMST.ct}, \texttt{AdaRMST.dt}, and \texttt{hulc}, across various scenarios. The summaries include the Bias, Monte Carlo Standard Deviation, and Root Mean Square Error (RMSE) relative to their respective true values.
In terms of RMSE, \texttt{AdaRMST.dt} consistently outperforms both \texttt{AdaRMST.ct} and \texttt{hulc} across sample sizes in scenarios such as \texttt{null}, \texttt{ph}, \texttt{early}, \texttt{delay\_1}, \texttt{delay\_2}, and \texttt{delaycon}. Conversely, \texttt{AdaRMST.ct} and \texttt{hulc} exhibit superior performance in the \texttt{tran}, \texttt{cs}, and \texttt{msep} scenarios.

\textbf{Recommendations} \quad Based on the above findings, we recommend the use of \texttt{AdaRMST.ct} for its balance of power and coverage, strong performance in both early and delayed treatment effect scenarios, coherence between estimation and testing, independence from grid selection, and well-controlled type I error.

\section{Application}
\label{sec:app}
We applied our method to a recent oncology clinical trial by \citet{tempero2023adjuvant}. This randomized, open-label phase III trial (ClinicalTrials.gov identifier: NCT01964430) evaluated the efficacy and safety of adjuvant nab-paclitaxel plus gemcitabine (nab-P + Gem) versus gemcitabine (Gem) alone in 866 patients with surgically resected pancreatic ductal adenocarcinoma. The primary endpoint, independently assessed disease-free survival (DFS), was not met, with a median DFS of 19.4 months for nab-P + Gem versus 18.8 months for Gem (hazard ratio [HR] = 0.88, 95\% CI: $[0.729, 1.063]$, $p$-value $=0.18$).

\begin{figure} 
  \centering \includegraphics[width=0.9\linewidth]{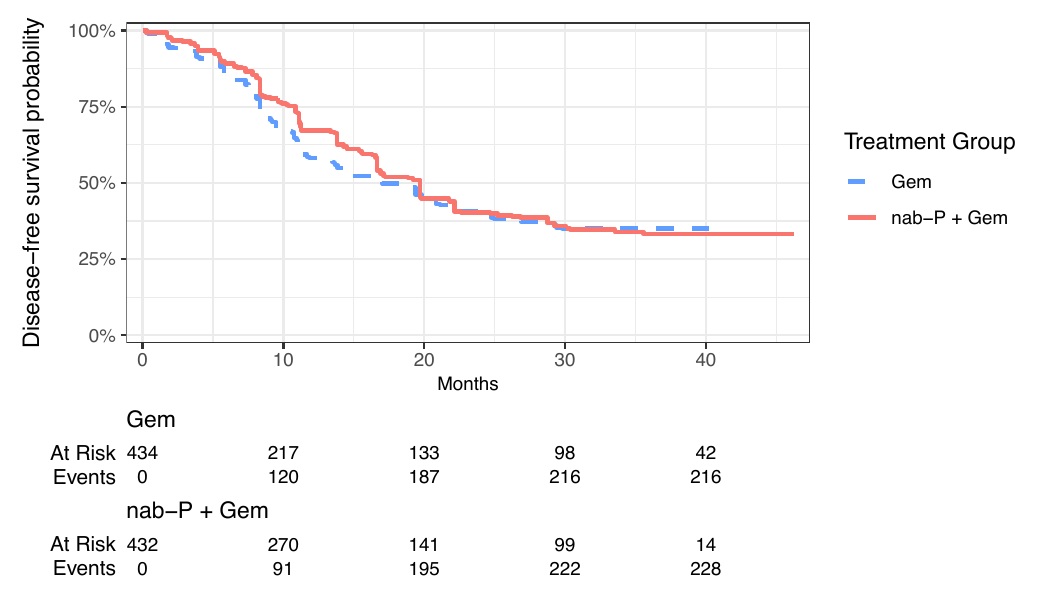} 
  \caption{Reconstructed Kaplan-Meier plot of disease-free survival from the phase III randomized trial on adjuvant therapy for pancreatic cancer by \citet{tempero2023adjuvant}.} 
  \label{fig:app} 
\end{figure}

\textbf{Motivation for Adaptive Methods} \quad The DFS curves from the trial (Figure~\ref{fig:app}, reconstructed) show a transient treatment effect, with the curves coinciding at approximately 18 months. This pattern aligns with the \texttt{tran} scenario in our simulation study, where a treatment effect is temporary and fades over time. Such a pattern violates the proportional hazards assumption, rendering the log-rank test underpowered. Despite a clear separation of survival curves in the first 18 months, the trial was deemed negative, failing to demonstrate sufficient evidence to change the standard of care. This highlights a limitation of relying on the log-rank test in settings with non-proportional hazards.

Given the difficulty in anticipating the treatment effect pattern in advance, an adaptive method that does not rely on the proportional hazards assumption, such as our proposed approach, may be more suitable. In the following, we reanalyze this trial using \texttt{AdaRMST.ct} to evaluate its performance and clinical implications.

\textbf{Data Reconstruction and Method Configuration} \quad We reconstructed the trial data from the Kaplan-Meier plot in Figure~2(A) of the original publication using the \texttt{IPDfromKM} \texttt{R} package \citep{liu2021ipdfromkm}, a method widely validated for secondary analysis of survival data. Although the reconstructed data may not perfectly match the original, the resulting Kaplan-Meier curves closely resemble the published plot, ensuring the validity of our analysis.

For \texttt{AdaRMST.ct}, our primary analysis method, without expert knowledge, we set the candidate restriction time range to $\mathcal{L}=[3,53]$ months, based on the trial's maximum follow-up time. We use $B = 6000$ bootstrap resamples. The initial guess is set as the default midpoint $\tilde{L} = 28$ months. The penalty parameter $c$ was determined using the default formula:
\[c = 0.002\times\ff{16}{(L_{\max} - L_{\min})^2}\times \left(\ff{1 \text{ month}}{1 \text{ year}}\right)^2 = 0.002\times\ff{16}{(53 - 3)^2}\times \left(\ff{1}{12}\right)^2 =  8.89\times 10^{-8}.\]

\textbf{Results} \quad We began by estimating the hazard ratio and performing the log-rank test on the reconstructed data. The log-rank test yielded a $p$-value of 0.40, and the estimated hazard ratio was $\text{HR} = 0.92$ (95\% CI: $[0.77, 1.11]$). These results align with the original publication but suggest a slightly smaller effect size.

The estimated optimal restriction time using \texttt{AdaRMST.ct} was $\hat{L}^\dagger = 15.5$ months (95\% CI: $[3.0, 19.7]$), effectively capturing the transient effect window. The estimated treatment effect was $\hat{\ka}^\dagger = 0.845$ months (95\% CI:  $[0.0345, 1.5256]$ months, $p$-value $=0.0392$), indicating a statistically significant benefit of nab-P + Gem over Gem at a significance level of $\alpha = 0.05$. This suggests that, on average, patients in the treatment group lived 0.845 months longer than those in the control group within the first 15.5 months of follow-up, representing a clinically meaningful outcome.

To provide a broader comparison, we also analyzed the data using \texttt{AdaRMST.dt}, \texttt{hulc}, \texttt{rmst}, \texttt{maxcombo}, and \texttt{hctu}. Detailed descriptions of these methods are provided in Section~\ref{sec:simulation}. The results are summarized as follows: 
(1) \texttt{AdaRMST.dt}: Using 10 equally spaced grid points in $[3, 53]$ months and the discrete-time default parameter $c = 0.005 \times \frac{16}{(L_{\max} - L_{\min})^2} \times \left(\frac{1 \text{ month}}{1 \text{ year}}\right)^2 = 2.22 \times 10^{-7}$, the estimates were $\hat{L}^\dagger = 14.1$ months and $\hat{\ka}^\dagger = 0.733$ months (95\% CI: $[0.195, 1.27]$, $p$-value $=7.56 \times 10^{-3}$). (2) \texttt{hulc}: The estimates were $\hat{L}^* = 15.4$ months and $\hat{\ka}^* = 0.843$ months (95\% CI: $[-2.30, 1.56]$). (3) \texttt{rmst}: At the default maximum follow-up restriction time of 41.3 months, the treatment effect was $\hat{\ka} = 0.964$ months (95\% CI: $[-1.25, 3.18]$, $p$-value: 0.393). (4) \texttt{maxcombo}: The method yielded a $p$-value of 0.196. (5) \texttt{hctu}: The estimates were $\hat{L}^* = 15.7$ months and $\hat{\ka}^* = 0.861$ months (95\% CI: $[0.096, 1.627]$, $p$-value $=0.02$).

\textbf{Conclusion} \quad Our reanalysis demonstrates that \texttt{AdaRMST.ct} effectively identifies a transient treatment effect that the log-rank test failed to detect and that may also have been overlooked by canonical RMST or \texttt{maxcombo} methods. Had adaptive RMST methods been employed, the trial could have been interpreted as positive, highlighting the potential early-phase benefit of nab-P + Gem. This underscores the value of adaptive methods in survival analysis, particularly for trials where the proportional hazards assumption is unlikely to hold.

\section{Discussion}
\label{sec:discussion}

Given the prevalence of non-PH scenarios in both clinical trials and real-world applications, it is imperative to develop methods that are broadly applicable and capable of adapting to diverse treatment effect patterns without requiring prior knowledge. 
To address this need, the proposed adaptive RMST-based methods offer a balance of flexibility, robustness, interpretability, statistical power, and coherence. Their adoption in clinical trials and real-world studies has the potential to significantly improve the detection and interpretation of treatment effects, ultimately benefiting both research and patient outcomes. Among these methods, the continuous-time penalized adaptive RMST approach (\texttt{AdaRMST.ct}) stands out as our recommended choice due to its strong and consistent performance across a variety of scenarios. 

While \texttt{AdaRMST.ct} requires specifying two parameters — $\tilde{L}$ (the initial guess for the restriction time) and $c$ (the penalty parameter) — our simulations and oncological trial example suggest that users can rely on default settings for routine applications. Specifically, a small penalty parameter, combined with a reasonable default choice for $\tilde{L}$, such as the midpoint of the candidate range, consistently yields reliable results, as illustrated by Figures \ref{fig:pen_300} - \ref{fig:pen_1000} in the Supplementary Material. This parameter configuration ensures ease of use while maintaining adaptability for a range of scenarios.

This work also opens several directions for future research: (1) Sample Size Calculation: Developing reliable and practical sample size formulas tailored to the proposed RMST-based methods would be invaluable for trial planning and ensuring adequately powered studies. 
(2) Incorporating Covariates: Extending the methods to account for baseline covariates would increase their applicability to more complex trial designs, allowing for greater precision and interpretability in stratified or adjusted analyses.

\section*{Acknowledgement}

The authors would like to acknowledge the National Institutes of Health (NIH) for their generous funding and support.

\section*{Supplementary Material}
\label{SM}

The Supplementary Material provides all technical proofs, supplementary results, and detailed descriptions of the simulations, including extended findings. The \texttt{R} package \texttt{AdaRMST} and reconstructed trial data referenced in Section~\ref{sec:app} are available online at \url{https://github.com/Jinghao-Sun/AdaRMST}.

\bibliographystyle{plainnat}
\bibliography{main}

\appendix

\renewcommand{\thefigure}{S\arabic{figure}}
\renewcommand{\thetable}{S\arabic{table}}
\renewcommand{\thetheorem}{S\arabic{theorem}}
\renewcommand{\thelemma}{S\arabic{lemma}}
\renewcommand{\theproposition}{S\arabic{proposition}}
\renewcommand{\theequation}{S\arabic{equation}}
\renewcommand{\thesection}{S\arabic{section}}

% Reset counters
\setcounter{figure}{0}
\setcounter{table}{0}
\setcounter{theorem}{0}
\setcounter{lemma}{0}
\setcounter{equation}{0}
\setcounter{section}{0}
\setcounter{proposition}{0}

\section{Proofs}

\subsection{Lemma \ref{lem:1}}

\begin{lemma}
  \label{lem:1}
  Under Assumptions \ref{assum:setting} and \ref{assum:smooth}, we have
  \[\mathbb{M}_n(L) \pto \mathbb{M}(L) \text{ uniformly on } \mathcal{L}, \text{ as } n \to \infty.\]
\end{lemma}
\begin{proof}
  $\ff{\hat\ka^2(L)}{\hat\sigma^2(\hat\ka(L))}$:  By the uniform consistency of KM estimator $\hat{S}^a(t)$ \citep{gill1983large}, $\hat{\theta}^a(L)$ is a uniformly consistent estimator of $\theta^a(L)$ \citep{zhao2016restricted}. Thus, $\hat\ka(L)$ is a uniformly consistent estimator of $\ka(L)$. We know that $\hat\sigma(\hat\ka(L))$ is a consistent estimator of the asymptotic standard deviation of $\hat\theta^1(L) - \hat\theta^0(L)$ \citep[Chap.~IV]{andersen1993statistical}. Therefore, $\ff{\hat\ka^2(L)}{\hat\sigma^2(\hat\ka(L))}$ is a consistent estimator of $\ff{\ka^2(L)}{\sigma^2(\hat\ka(L))}$ by the Continuous Mapping Theorem. 

  Note that $\mathbb{M}(L)$ and $\mathbb{M}_n(L)$ are continuous in $L$, since both their numerator and denominator are continuous in $L$. Since $\mathcal{L}$ is compact, we know that $\mathbb{M}$ and $\mathbb{M}_n$ are continuous and bounded functions on $\mathcal{L}$, i.e. $\ell^{\infty}(\mathcal{L})$. Next, we show that $\mathbb{M}_n$ is uniformly consistent for $\mathbb{M}$ on $\mathcal{L}$. We do so with an equicontinuity argument.

  Given a positive number $\delta > 0$, and $s,t \in \mathcal{L} \equiv [L_{\min}, L_{\max}]$ with $s<t$, we define
  \[\Delta \hka_{s,t} = \hka_t - \hka_s, 
  \quad \Delta \hat\sigma^2_{s,t}(\hat\ka) = \hat\sigma^2(\hat\ka_t) - \hat\sigma^2(\hat\ka_s) \ge 0.\]
  We have the following inequalities:
  \begin{equation*}
    \begin{aligned}
    & -s \le \hka_s \le s,  \quad  s-t \le\Delta \hka_{s,t} \le t-s,\\
    & 0 \le \Delta \hat\sigma^2_{s,t}(\hat\ka) \le n(t^2-s^2) \sum_{a =0}^1 \sum_{j:T_{a(j)} \le s} \ff{d_{a(j)}}{\left[Y^a(T_{a(j)})\right]\left[Y^a(T_{a(j)} )- d_{a(j)}\right]}\\
    &\quad + n(t-s)^2 \sum_{a =0}^1 \sum_{j: s< T_{a(j)} \le t} \ff{d_{a(j)}}{\left[Y^a(T_{a(j)})\right]\left[Y^a(T_{a(j)} )- d_{a(j)}\right]}.
    \end{aligned}
  \end{equation*}
  Then, we have modulus of continuity 
  \begin{equation*}
  \begin{aligned}
    &\sup_{t -s <\delta}| \mathbb{M}_n(t) - \mathbb{M}_n(s)|\\ 
    &= \sup_{t -s <\delta} \left|\ff{\hat\ka^2_t}{\hat\sigma^2(\hat\ka_t)} - \ff{\hat\ka^2_s}{\hat\sigma^2(\hat\ka_s)}\right| \\
    &= \sup_{t -s <\delta} \left|\ff{2\Delta \hka_{s,t}\hka_s \hat\sigma^2(\hat\ka_s) + \Delta \hka_{s,t}^2 \hat\sigma^2(\hat\ka_s) -  \hka_s^2 \Delta \hat\sigma^2_{s,t}(\hat\ka)}
    {\Delta\hat\sigma^2_{s,t}(\hat\ka)\hat\sigma^2(\hat\ka_s) + \hat\sigma^4(\hat\ka_s)} \right| \\
    & \le \sup_{t -s <\delta}  \left|\ff{2(t-s)s  + (t-s)^2  }
    {\Delta\hat\sigma^2_{s,t}(\hat\ka) + \hat\sigma^2(\hat\ka_s)} \right|
    +  \left|\ff{s^2}
    {\hat\sigma^2(\hat\ka_s) + \ff{\hat\sigma^4(\hat\ka_s)}{\Delta \hat\sigma^2_{s,t}(\hat\ka)}} \right| \\
    &\le \ff{2\delta L_{\max}  + \delta^2  }
    {\hat\sigma^2(\hat\ka_{L_{\min}})} + 
    \ff{L_{\max}^2}
    {\hat\sigma^2(\hat\ka_{L_{\min}}) + \ff{\hat\sigma^4(\hat\ka_{L_{\min}})}{(\delta^2 + 2\delta L_{\max})n\sum_{a =0}^1 \sum_{j:T_{a(j)} \le L_{\max}} \ff{d_{a(j)}}{\left[Y^a(T_{a(j)})\right]\left[Y^a(T_{a(j)} )- d_{a(j)}\right]}}} 
  \end{aligned}
  \end{equation*}
  Since $\hat\sigma^2(\hat\ka_{L_{\min}}) \pto \sigma^2(\hat\ka_{L_{\min}}) > 0$, and $n\sum_{a =0}^1 \sum_{j:T_{a(j)} \le L_{\max}} \ff{d_{a(j)}}{\left[Y^a(T_{a(j)})\right]\left[Y^a(T_{a(j)} )- d_{a(j)}\right]} $ is a consistent estimator for the asymptotic variance of (weighted) sum of the Nelson-Aalen estimators of the cumulative hazard functions for each arm and therefore is much greater than zero \citep[Chap. 3]{aalen2008survival}, we have that 
  \[ \lim_{\delta \to 0}\limsup_{n\to \infty} \Pr\left( \sup_{t -s <\delta}| \mathbb{M}_n(t) - \mathbb{M}_n(s)| >\epsilon \right) = 0.\]
  Therefore, invoking Theorem 18.14 in \citet{van2000asymptotic}, we have the uniform consistency of $\mathbb{M}_n$.
\end{proof}

\subsection{Lemma \ref{lem:var}}
We define below the estimands and estimators related to the variance of RMST estimator and its derivative. For the ease of notation, we describe the generic case and therefore omit the superscript $a$. Recall that
$\alpha_u \equiv \ff{f_u}{S_u}$,
$y_u \equiv \Pr(X \ge u)= S_u \G_u = \Pr(T \ge u)\Pr(C \ge u)$,
$\sigma^2_s = \int_0^s \ff{\alpha_u}{y_u} \dy{u}$.
Define $N_s = \sum_{i = 1}^n \I{X_i \le s, \Delta_i = 1},  Y_s = \sum_{i = 1}^n\I{X_i \ge s}$, $J_s \equiv \I{Y_s > 0}$. It is common to focus on time $s$ when $J_s = 1$, therefore, without loss of generality, we omit it in the formulae in this paper. Recall that the asymptotic variance and its derivative are \citep[Chap.~IV.3]{andersen1993statistical}
\begin{equation*}
  \begin{aligned}
    V_L &\equiv \sigma^2(\hat\theta(L)) = \int_0^L (\theta_L - \theta_v)^2 \ff{\alpha_v}{y_v} \dy{v},\\
    \dot{V}_L &= 2S_L\int_0^L (\theta_L - \theta_v)\ff{\alpha_v}{y_v}\dy{v} = 2S_L\theta_L\int_0^L \ff{\alpha_v}{y_v}\dy{v} - 2S_L\int_0^L  \theta_v\ff{\alpha_v}{y_v}\dy{v}.
  \end{aligned}
\end{equation*}
By Theorem IV.3.2 in \citet{andersen1993statistical}, both $\hat{\sigma}^2_t \equiv n\int_0^t \ff{\dy{N_s}}{Y_s^2}$ and the alternative formula $\hat{\sigma}^{2, \text{Greenwood}}_t \equiv n\int_0^t \ff{\dy{N_s}}{Y_s(Y_s - \Delta N_s)}$ associated with the Greenwood's formula are uniformly consistent for $\sigma^2_t$ and thus also for each other. We use the Greenwood's formula version plug-in in $\mathbb{M}_n$ for its better finite sample properties. However, the asymptotic properties of our estimators are the same plugging in either estimator for $\sigma^2_t$. Therefore, in the following, to simplify notation, we only show the result with $\hat{\sigma}^2_t$. Define the estimators as
\begin{equation*}
  \begin{aligned}
    \hat{V}_L &\equiv \hat\sigma^2(\hat\theta(L)) 
    = \int_0^L (\hat\theta_L - \hat\theta_v)^2 \dy{\hat{\sigma}^2_v} 
    = \int_0^L \ff{n(\hat\theta_L - \hat\theta_v)^2 }{Y_v^2}\dy{N_v},\\
    \hat{\dot{V}}_L &= 2\hat{S}_L\int_0^L \ff{n(\hat\theta_L - \hat\theta_v)}{Y_v^2}\dy{N_v}.
  \end{aligned}
\end{equation*}
In the next result, we prove the asymptotic normality of RMST variance estimator and its derivative estimator. The consistency of the variance estimator is well known from the literature, e.g. \citep{andersen1993statistical}, while its asymptotic distribution is not yet established to the best of our knowledge. We further prove the consistency and asymptotic normality of the derivative estimator, which is needed for the subsequent proofs.
\begin{lemma}[Asymptotic Properties of RMST variance estimator and its derivative]
  \label{lem:var}
  Under Assumptions \ref{assum:setting}, \ref{assum:smooth}, as $n \to \infty$, we have $\hat{V}_L \pto {V}_L, \hat{\dot{V}}_L \pto {\dot{V}}_L$, and 
  \begin{equation*}
  \begin{aligned}
    \sqrt{n}(\hat{V}_L - V_L) &\dto \mathcal{N}(0, \sigma^{V,2}_L), \\
    \sqrt{n}(\hat{\dot{V}}_L - \dot{V}_L) &\dto \mathcal{N}(0, \sigma^{\dot{V},2}_L).
  \end{aligned}
  \end{equation*}
\end{lemma}

\begin{proof}
  Recall that the estimators are defined as
  \begin{equation*}
  \begin{aligned}
    \hat{S_t} &= \prod_{s\le t} \left(1 - \ff{\Delta N_s}{Y_s} \right)\\
    &\text{(By IV.3.1 in \citet{andersen1993statistical}, this KM is uniformly consistent for $S_t$.)}\\
    \hat{\sigma}^2_L &= n\int_0^L \ff{1}{Y_s^2}\dy{N_s}\\
    &\text{(By 4.3.18 in \citet{andersen1993statistical}, $\hat{\sigma}^2_L$ is uniformly consistent for $\sigma^2_L$.)}\\
    %\hat{\sigma}^{2,GW}_L &= \int_0^L \ff{1}{Y_s(Y_s - \Delta N_s)}\dy{N_s} \text{ (Greenwoods's formula)}\\
    \hat{V}_L &\equiv \hat\sigma^2(\hat\theta(L)) 
    = \int_0^L (\hat\theta_L - \hat\theta_v)^2 \dy{\hat{\sigma}^2_v} 
    = \int_0^L \ff{n(\hat\theta_L - \hat\theta_v)^2 }{Y_v^2}\dy{N_v},\\
    \hat{\dot{V}}_L &=  2\hat{S}_L\int_0^L (\hat\theta_L - \hat\theta_v) \dy{\hat{\sigma}^2_v} 
    = 2\hat{S}_L\int_0^L \ff{n(\hat\theta_L - \hat\theta_v)}{Y_v^2}\dy{N_v}.
  \end{aligned}
  \end{equation*}
  The consistency of $\hat{V}_L$ is known in the literature, e.g. \citep{andersen1993statistical}. The consistency of $\hat{\dot{V}}_L$ can be proved as follows:\\
  Since $\hat{\sigma}^2_L$ is uniformly consistent for $\sigma^2_L$, and $(\hat\theta_L - \hat\theta_v)$ is also uniformly consistent for $(\theta_L - \theta_v)$, we have the uniformly consistency of $\hat{\dot{V}}_L$ by the Continuous Mapping Theorem and the continuity of the integration operator. 

  We next decompose the above into simpler terms. 
  We have
  \begin{equation*}
  \begin{aligned}
    \hat{V}_L &
    = \int_0^L \ff{n(\hat\theta_L - \hat\theta_v)^2 }{Y_v^2}\dy{N_v} = \hat\theta_L^2\int_0^L \ff{n}{Y_v^2}\dy{N_v} + \int_0^L \ff{n \hat\theta_v^2 }{Y_v^2}\dy{N_v} - 2\hat\theta_L\int_0^L \ff{n\hat\theta_v }{Y_v^2}\dy{N_v},\\
    \hat{\dot{V}}_L &= 2\hat{S}_L\int_0^L \ff{n(\hat\theta_L - \hat\theta_v)}{Y_v^2}\dy{N_v}= 2\hat{S}_L\hat\theta_L\int_0^L \ff{n}{Y_v^2}\dy{N_v} - 2\hat{S}_L\int_0^L \ff{n\hat\theta_v}{Y_v^2}\dy{N_v}
  \end{aligned}
  \end{equation*}

  Since $\hat{S}_L, \hat{\theta}_L, \hat{\theta}_L^2$ are asymptotically normal, by delta method, to prove the asymptotic Normality of $\hat{V}_L$ and $\hat{\dot{V}}_L$, we only need to focus on
  \[\int_0^L \ff{n}{Y_v^2}\dy{N_v} \to \int_0^L \ff{\alpha_v}{y_v}\dy{v}, \int_0^L \ff{n\hat\theta_v}{Y_v^2}\dy{N_v} \to  \int_0^L  \theta_v\ff{\alpha_v}{y_v}\dy{v},\int_0^L \ff{n \hat\theta_v^2 }{Y_v^2}\dy{N_v} \to \int_0^L  \theta_v^2\ff{\alpha_v}{y_v}\dy{v}.\]
  By $\dy{N}_v = Y_v\alpha_v\dy{v} + \dy{M}_v,$ we have
  \begin{equation*}
  \begin{aligned}
  \sqrt{n}\left\{\int_0^L\ff{n\dy{N}_v}{Y_v^2} - \int_0^L  \ff{\alpha_v}{y_v}\dy{v}\right\} 
  &= \sqrt{n}\left\{\int_0^L\ff{n\dy{N}_v}{Y_v^2} - \int_0^L  \ff{n\alpha_v}{Y_v}\dy{v}\right\} 
  + \sqrt{n}\left\{\int_0^L  \ff{n\alpha_v}{Y_v}\dy{v} - \int_0^L  \ff{\alpha_v}{y_v}\dy{v}\right\}\\
  &= \sqrt{n}\int_0^L\ff{n\dy{M}_v}{Y_v^2} 
  + \int_0^L \alpha_v \sqrt{n}\left(\ff{1}{Y_v/n} - \ff{1}{y_v}\right)\dy{v}\\
  \sqrt{n}\left\{\int_0^L\ff{n \hat{\theta}_v\dy{N}_v}{Y_v^2} 
  - \int_0^L  \ff{{\theta}_v\alpha_v}{y_v}\dy{v}\right\} 
  &= \sqrt{n}\int_0^L\ff{n\hat{\theta}_v\dy{M}_v}{Y_v^2} 
  + \int_0^L \alpha_v \sqrt{n}\left(\ff{\hat{\theta}_v}{Y_v/n} - \ff{{\theta}_v}{y_v}\right)\dy{v}\\
  \sqrt{n}\left\{\int_0^L\ff{n \hat{\theta}^2_v\dy{N}_v}{Y_v^2} 
  - \int_0^L  \ff{{\theta}^2_v\alpha_v}{y_v}\dy{v}\right\} 
  &= \sqrt{n}\int_0^L\ff{n\hat{\theta}^2_v\dy{M}_v}{Y_v^2} 
  + \int_0^L \alpha_v \sqrt{n}\left(\ff{\hat{\theta}^2_v}{Y_v/n} - \ff{{\theta}^2_v}{y_v}\right)\dy{v}
  \end{aligned}
  \end{equation*}
  The first term on the right-hand side in each of the above three equations is Gaussian by Rebolledo's Martingale Central Limit Theorem (MCLT) (Theorem II.5.1 in \citet{andersen1993statistical}) while the second term is also Gaussian by Functional Delta Method (FDM) and Continuous Mapping Theorem. Therefore, we have the asymptotic normality of the above three terms on the left-hand side. 
  In the following, we only show the detailed derivation for $\sqrt{n}\int_0^L\ff{n\hat{\theta}_v\dy{M}_v}{Y_v^2}$ and 
  $\int_0^L \alpha_v \sqrt{n}\left(\ff{\hat{\theta}_v}{Y_v/n} - \ff{{\theta}_v}{y_v}\right)\dy{v}$ and the other two are analogous and thus omitted.

  (i) MCLT:\\
  Define $M^n_t \equiv \sqrt{n}\int_0^t\ff{n\hat{\theta}_v\dy{M}_v}{Y_v^2}$. Then,
  \[\langle M^n \rangle_t = \int_0^t \ff{n^3\hat{\theta}_v^2}{Y_v^4} \dy{\langle M\rangle}_v =  \int_0^t \ff{n^3\hat{\theta}_v^2}{Y_v^4} Y_v \alpha_v \dy{v} = \int_0^t \ff{\alpha_v\hat{\theta}_v^2}{(Y_v/n)^3} \dy{v}, \]
  \[\langle M^{n,\epsilon}\rangle_t =  \int_0^t \ff{n^3\hat{\theta}_v^2}{Y_v^4} \I{\left|\sqrt{n} \ff{n\hat{\theta}_v}{Y_v^2} \right|> \epsilon } \dy{\langle M\rangle}_v = \int_0^t \ff{\alpha_v\hat{\theta}_v^2}{(Y_v/n)^3} \I{\left|\ff{1}{\sqrt{n}} \ff{\hat{\theta}_v}{(Y_v/n)^2} \right|> \epsilon } \dy{v}.\]
  The integrand in the first equation converges to a fixed quantity while the second equation converges to zero. Therefore, we only need to find a dominating integrable function for both to invoke Dominated Convergence Theorem or Proposition II.5.3 in \citet{andersen1993statistical}. In particular, we need to find positive function $k_v^{\delta}$ such that $\int_0^{L_{\max}} k_v^{\delta} \dy{v} < \infty$ and
  \[\liminf_{n \to \infty} \Pr\left(\left| \ff{\alpha_v\hat{\theta}_v^2}{(Y_v/n)^3}\right| \le k_v^{\delta} \text{ for all $v \in [0, L]$}  \right) \ge 1 - \delta.\]
  Define ${H}^n$ and $H$ as the empirical and true CDF of follow-up time $X_i$. Thus, we have 
  \begin{equation*}
  \begin{aligned}
    &\Pr\left(\left| \ff{\alpha_v\hat{\theta}_v^2}{(Y_v/n)^3}\right| \le k(v,\delta)  \text{ for all $v \in [0, L]$}\right)\\ 
    & \text{(since when $v = 0$, it is always true, we omit it below)}\\
    &\ge \Pr\left(\ff{1}{Y_v/n}\le \left(\ff{k(v,\delta)}{\alpha_v v^2}\right)^{1/3}  \text{ for all $v \in (0, L]$}\right)\\
    &= \Pr\left( {H}^n_v \le 1 - \left(\ff{\alpha_v v^2}{k(v,\delta)}\right)^{1/3}  \text{ for all $v \in (0, L]$}\right)
  \end{aligned}
  \end{equation*}
  Therefore, choosing the integrable function $k(v,\delta) = \ff{8\alpha_v v^2}{(1 - H_L)^3}$, we have 
  $1 - \left(\ff{\alpha_v v^2}{k(v,\delta)}\right)^{1/3} =  \ff{1 + H_L}{2}.$ Then, we have 
  \begin{equation*}
  \begin{aligned}
    &\liminf_{n \to \infty}\Pr\left(\left| \ff{\alpha_v\hat{\theta}_v^2}{(Y_v/n)^3}\right| \le k(v,\delta) \text{ for all $v \in [0, L]$}  \right)\\ 
    &\ge \liminf_{n\to \infty}\Pr\left( {H}^n_v \le \ff{1 + H_L}{2} \text{ for all $v \in (0, L]$}\right)\\
    & = \liminf_{n\to \infty}\Pr\left( {H}^n_L \le \ff{1 + H_L}{2}\right)\\
    & = 1 > 1- \delta.
  \end{aligned}
  \end{equation*}
Therefore, 
\[\langle M^n \rangle_t \pto V_t = \int_0^t \ff{\alpha_v\theta_v^2}{y_v^3}\dy{v},\]
\[\langle M^{n,\epsilon} \rangle_t \pto 0.\]
Hence, by MCLT, we have $M^n$ converges weakly to a mean-zero Gaussian martingale with covariance function $V_t$.

(ii) FDM:\\
Within the second term $\int_0^L \alpha_v \sqrt{n}\left(\ff{\hat{\theta}_v}{Y_v/n} - \ff{{\theta}_v}{y_v}\right)\dy{v}$, we have $ \sqrt{n}(\hat{\theta} - {\theta})$ and $\sqrt{n}(\ff{Y}{n} - y)$ both converge weakly to Gaussian processes. Therefore, since ratio is a tangentially compactly differentiable map, by Functional Delta Method, we have the weak convergence of $\alpha_v \sqrt{n}\left(\ff{\hat{\theta}_v}{Y_v/n} - \ff{{\theta}_v}{y_v}\right)$ to a mean-zero Gaussian process. Thus, since the integral of a Gaussian process is also a Gaussian process, and by the continuous mapping theorem, $\int_0^L \alpha_v \sqrt{n}\left(\ff{\hat{\theta}_v}{Y_v/n} - \ff{{\theta}_v}{y_v}\right)\dy{v}$ is asymptotically normal. Hence, we have finished the proof.
\end{proof}

\subsection{Proposition \ref{thm:can}}
Consistency:
\begin{proof}
  (1) $\hL^*$: Under Assumption \ref{assum:unique}, and the uniform consistency by Lemma \ref{lem:1}, by Argmax Continuous Mapping Theorem \citep[Lemma 3.2.1]{vandervaart1996weak}, we have $\hL^* \pto L^*$.
  
  (2) $\hka^*$: We have 
  \begin{equation*}
    \begin{aligned}
      |\hka^* - \ka^*| &\le |\hka(\hL^*) - \ka(\hL^*)| + |\ka(\hL^*) - \ka(L^*)| \\
      &\le \sup_{L \in \mathcal{L}}|\hka(L) - \ka(L)| + |\ka(\hL^*) - \ka(L^*)|.
    \end{aligned}
  \end{equation*}
  Then, by the uniform consistency of $\hat\ka(L)$ on $\mathcal{L}$ by Lemma~\ref{lem:1}, as well as the consistency of $\hL^*$ and the continuity of $\ka(\cdot)$, we have $\hka^* \pto \ka^*$.
\end{proof}
Asymptotic Normality:
\begin{proof}
  (I) $\hL^*$: We will invoke Theorem 3.2.16 in \citet{vandervaart1996weak}. 
  
  (i) We first verify that Assumption~\ref{assum:smooth} ensures that $\mathbb{M}$ is twice continuously differentiable at $L^*$. Recall that $\mathbb{M}_L = \ff{\ka_L^2}{\sigma^{\ka,2}_L}$. Therefore, it suffices to show $\theta_L$ and $\sigma^2(\hat\theta^a(L))$ are twice continuously differentiable at $L^*$. Since $\theta^a(L) = \int_0^L S^a_u \dy{u}$, therefore, $\partial_L^2\theta^a(L) = -f^a(L)$, which is continuous at $L^*$. By \citet[Chap. IV.3]{andersen1993statistical}, denote \[V^a_L \equiv \sigma^2(\hat\theta^a(L)) = \int_0^L (\theta^a_L - \theta^a_v)^2 \ff{\alpha_v^a}{y_v^a} \dy{v},\]
  where $\alpha^a$ and $y^a$ are defined in Lemma \ref{lem:var}. % or in Theorem IV.3.2 in \citet{andersen1993statistical}. 
  Then, we have
  \begin{equation*}
    \begin{aligned}
      \partial_L V^a_L &= 2\int_0^L (\theta^a_L - \theta_v^a)S^a_L\ff{\alpha_v^a}{y_v^a}\dy{v},\\
      \partial_L^2 V^a_L &= 2\int_0^L \left[ (S_L^a)^2 - f_L^a\cdot(\theta_L^a - \theta_v^a) \right]\ff{\alpha_v^a}{y_v^a}\dy{v}.
    \end{aligned}
  \end{equation*}
  Thus, $\partial_L^2 V^a_L$ exists and continuous at $L^*$ by Assumption~\ref{assum:smooth}.

  (ii) Next, we prove that $\sqrt{n}(\mathbb{M}_n - \mathbb{M})$ converges weakly to a mean-zero Gaussian process $G^M$ with $\Sigma^M(\cdot,\cdot)$ the covariance function. Note that
  \begin{equation*}
    \begin{aligned}
      \sqrt{n}(\mathbb{M}_n(L) - \mathbb{M}(L)) &= \sqrt{n}\left(\ff{\hat\ka^2(L)}{\hat\sigma^2(\hat\ka(L))} - \ff{\ka^2(L)}{\sigma^2(\hat\ka(L))}\right),
    \end{aligned}
  \end{equation*}
  and by \citet[Chap. IV.3]{andersen1993statistical}, we have weak convergence on $\mathcal{L}$
\begin{equation}
  \sqrt{n}(\hka_{\cdot} - \ka_{\cdot}) \dto  \GP^\ka \sim - \ff{1}{\sqrt{\beta}} \int_0^{\cdot}\int_v^{\cdot} S^1(u)\dy{u}\dy{U^1_v} + \ff{1}{\sqrt{1-\beta}}\int_0^{\cdot}\int_v^{\cdot} S^0(u)\dy{u}\dy{U^0_v},
  \label{eq:kappaasymp}
\end{equation}
where $U^1 \indep U^0$, and $\GP^\ka$ is a mean-zero Gaussian process with covariance function $\Sigma^\ka(\cdot,\cdot)$, and by Lemma~\ref{lem:var}, we have \text{(we can only have this at $\ka \ne 0$, it is the boundary)}
\begin{equation}
  \sqrt{n}(\hat{\sigma}^{\ka, 2}_L - {\sigma}^{\ka, 2}_L) \dto Z_L^{\sigma} \sim \mathcal{N}(0, \zeta_L).
  \label{eq:var_kappa} 
\end{equation}
  Note that $\sigma_L^{\ka, 2} = \GP^\ka(L,L)$.
  \footnote{%
  \parbox{0.95\linewidth}{
  Some auxiliary results not used in this proof: In a neighborhood of $L^*$, which by Assumption \ref{assum:unique} ensures $\ka(t) \ne 0$,  by the Continuous Mapping Theorem and the Delta method, we have the asymptotic distribution of $\sqrt{n}(\mathbb{M}_n - \mathbb{M})$ given by the following mean-zero Gaussian process $\GP^M$
  \begin{equation}
    \sqrt{n}(\mathbb{M}_n(L) - \mathbb{M}(L)) \dto \ff{2\ka_L}{\sigma^{\ka,2}_L}{\GP^\ka_L} - \ff{\ka^2_L}{\sigma^{\ka, 4}_L}Z^{\sigma}_L,
    \label{eq:asympMn}
  \end{equation}
  whose covariance function is
  \begin{equation*}
    \begin{aligned}
      \Sigma^M(s, t) &= \frac{4\kappa_s \kappa_t}{\sigma^{\kappa,2}_s \sigma^{\kappa,2}_t} \Sigma^{\ka}(s,t) + \frac{\kappa_s^2 \kappa_t^2}{\sigma^{\kappa,4}_s \sigma^{\kappa,4}_t} \text{Cov}(Z^{\sigma}_s, Z^{\sigma}_t) - \frac{2\kappa_s \kappa_t^2}{\sigma^{\kappa,2}_s \sigma^{\kappa,4}_t} \rho({s,t}) - \frac{2\kappa_t \kappa_s^2}{\sigma^{\kappa,2}_t \sigma^{\kappa,4}_s} \rho({t,s}).
    \end{aligned}
  \end{equation*}
  where $\rho(s,t) \equiv \Cov(\GP^\ka_s, Z^\sigma_t)$.
  }%
  }
  
  (iii) Next, we compute the derivatives of $\mathbb{M} - \mathbb{M}_n$ and related quantities with respect the time index. We list the results below:
  \begin{equation*}
    \begin{aligned}
      \dot{\ka}_t &= S^1_t - S^0_t, \quad
      \hat{\dot{\ka}}_t = \hat{S}^1_t - \hat{S}^0_t\\
      \dot{\sigma}^{\ka, 2}_t &= 
      \ff{2}{\beta} \int_0^t (\theta^1_t - \theta_v^1)S^1_t\ff{\alpha_v^1}{y_v^1}\dy{v} 
      + \ff{2}{1-\beta} \int_0^t (\theta^0_t - \theta_v^0)S^0_t\ff{\alpha_v^0}{y_v^0}\dy{v}\\
      \dot{\hat{\sigma}}^{\ka, 2}_t &=
      \sum_{a =0}^1 \sum_{j:T_{a(j)} \le t} \ff{2d_{a(j)}\hat{S}^a_t\left[\hat{\theta}^a(t) - \hat{\theta}^a\left(T_{a\left(j\right)}\right)\right]}{\left[Y^a(T_{a(j)})\right]\left[Y^a(T_{a(j)} )- d_{a(j)}\right]} \ a.s.\\
      \dot{\mathbb{M}}_t &
      = \ff{2\ka_t{\sigma}^{\ka, 2}_t\dot{\ka}_t - \ka^2_t\dot{\sigma}^{\ka, 2}_t}{{\sigma}^{\ka, 4}_t} , \quad
      \dot{\mathbb{M}}_{n,t} = 
      \ff{2\hat{\ka}_t\hat{\sigma}^{\ka, 2}_t\hat{\dot{\ka}}_t - \hat{\ka}^2_t\hat{\dot{\sigma}}^{\ka, 2}_t}{\hat{\sigma}^{\ka, 4}_t}
    \end{aligned}
  \end{equation*}
  Then, by Lemma~\ref{lem:var} and the delta method, in the neighborhood of $L^*$, we have
  \[\sqrt{n}( \dot{\mathbb{M}}_{n,t} - \dot{\mathbb{M}}_t) \dto \mathcal{N}(0, \sigma^{\dot{M},2}_t).\]
  Therefore, we invoke Theorem 3.2.16 in \citet{vandervaart1996weak} to conclude that 
  \[\sqrt{n}(\hL^* - L^*) \dto \mathcal{N}\left(0, \ff{\sigma^{\dot{M},2}_{L^*}}{[\dy^2_L \mathbb{M}(L^*)]^2}\right).\]
  
 (II) $\hka^*:$
  Clearly, $\hka$ and $\ka$ are differentiable functions of $L$. Then, we have
  \begin{equation*}
    \begin{aligned}
      \hka(\hL^*) &= \hka(L^*) + \hka'(L^*)(\hL^* - L^*) + o_p(|\hL^* - L^*|)\\
      &= \hka(L^*) + \hka'(L^*)(\hL^* - L^*) + o_p(n^{-1/2})\\
      \sqrt{n}(\hka(\hL^*) - \ka(L^*)) &= \sqrt{n}(\hka(L^*) - \ka(L^*)) + \sqrt{n}\hka'(L^*)(\hL^* - L^*) + o_p(1)
    \end{aligned}
  \end{equation*}
  Since $\hka$ converges to a Gaussian process, the first term is a normal random variable. Since $\hL^*$ is also asymptotically normally distributed, then, the RHS is just the sum of two normal random variables. Therefore, $\hka^*$ is asymptotically normal
  \[\sqrt{n}(\hka^* - \ka^*) \to \mathcal{N}(0, \sigma^{\ka^*,2}).\]
\end{proof}

\subsection{Proposition \ref{prop:consistency}}
\begin{proof}
  Recall that $\mathcal{L} = [L_{\min}, L_{\max}]$. Then, under $\mathsf{H}_0$, we have 
\begin{equation*}
  \begin{aligned}
    \sqrt{n}|\hka(\hL^*) - \ka(L^*)| & =   \sqrt{n}\left|\int_0^{\hL^*}\hat{S}^{1}_t - \hat{S}^{0}_t \dy{t}  \right|\\
    &=  \sqrt{n}\left|\int_0^{\hL^*} S^1_t + \hat{S}^{1}_t - S^1_t \dy{t} - \int_0^{\hL^*} {S}^{0}_t + \hat{S}^{0}_t - {S}^{0}_t \dy{t}  \right|\\
    &=  \sqrt{n}\left|\int_0^{\hL^*} S^1_t -S^0_t \dy{t} +\int_0^{\hL^*} \hat{S}^{1}_t - S^1_t \dy{t} - \int_0^{\hL^*} \hat{S}^{0}_t - {S}^{0}_t \dy{t}  \right|\\
    & =  \sqrt{n}\left|\int_0^{\hL^*} \hat{S}^{1}_t - S^1_t \dy{t} - \int_0^{\hL^*} \hat{S}^{0}_t - {S}^{0}_t \dy{t}  \right|\\
    & \le \sqrt{n}\left|\int_0^{\hL^*} \hat{S}^{1}_t - S^1_t \dy{t} \right| +   \sqrt{n}\left|\int_0^{\hL^*} \hat{S}^{0}_t - {S}^{0}_t \dy{t}  \right|\\
    & \le \sqrt{n}\int_0^{\hL^*} \left| \hat{S}^{1}_t - S^1_t \right| \dy{t} +  \sqrt{n}\int_0^{\hL^*}  \left|\hat{S}^{0}_t - {S}^{0}_t  \right|\dy{t}\\
    & \le L_{\max}\sup_{t \in [0, L_{\max}]} |\sqrt{n}(\hat{S}^{1}_t - S^1_t)| + L_{\max}\sup_{t \in [0, L_{\max}]} |\sqrt{n}(\hat{S}^{0}_t - S^0_t)|.
  \end{aligned}
\end{equation*}
Since $\sqrt{n}(\hat{S}^{a}_t - S^a_t)$ is a Gaussian process on a finite interval, $L_{\max}\sup_{t \in [0, L_{\max}]} |\sqrt{n}(\hat{S}^{a}_t - S^a_t)|$ is $O_p(1)$. Therefore, $|\hka(\hL^*) - \ka(L^*)| $ is $O_p(n^{-1/2})$, i.e. $\sqrt{n}$-consistent.
\end{proof}

\subsection{Proposition \ref{prop:median}}

\begin{proof}
We use a symmetry argument that induced by $S^1 = S^0$ to show that the asymptotic median of the sample distribution of $\hka^*$ is zero, which is just $\ka^*$, our estimand. Define process $H_{nt} \equiv \ff{\hka_t}{\hat{\sigma}_t}$. Under $\mathsf{H}_0$, 
\[\sqrt{n}H_{nt} = \ff{\sqrt{n}\hka_t}{\hat{\sigma}_t^{\ka}} \dto \ff{\mathcal{N}(0, \sigma^{\ka,2}_t)}{\sigma^{\ka}_t} = \mathcal{N}(0,1).\]
Note that $\mathbb{M}_{nt} = H_{nt}^2$, and therefore
\[\hat{L}^* = \argmax_{t\in\mathcal{L}} |H_{nt}|.\]
Define $\hat{H}^* \equiv H_{n\hat{L}^*}$. Since $H_{nt}$ is a mean-zero Gaussian process with constant variance, by its symmetry property, $\hat{H}^*$ is symmetrically distributed around 0, both marginally and conditionally given $\hat{L}^*$. Then, conditional on $\hat{L}^* = t$ for arbitrary $t \in \mathcal{L}$,
\begin{equation*}
    \begin{aligned}
    \sqrt{n}\hka^* &= \sqrt{n}H_{nt} \hat{\sigma}_{\hat{L}^*}\\
    &=  \sqrt{n}H_{nt} (\hat{\sigma}_{\hat{L}^*} - {\sigma}_t + {\sigma}_t)\\
    &=  \sqrt{n}H_{nt} (\hat{\sigma}_{\hat{L}^*} - {\sigma}_t) +  \sqrt{n}H_{nt}{\sigma}_t\\
    &=\sqrt{n}H_{nt} (\hat{\sigma}_{t} - {\sigma}_t) +  \sqrt{n}H_{nt}{\sigma}_t\\
    &= \sqrt{n}H_{nt}{\sigma}_t + o_p(1)
    \end{aligned}
\end{equation*} 
Therefore, by the conditional symmetry of $\sqrt{n}H_{nt}$ given $\hat{L}^*$, $\sqrt{n}\hka^*$ is also conditional symmetric around 0, and therefore marginally symmetric. This implies the asymptotic median of $\sqrt{n}\hka^*$ is 0, equal to $\ka^*$.
\end{proof}

\subsection{Proposition \ref{prop:hulc}}

\begin{proof}
  By Proposition~\ref{thm:can}, $\hka^*$ is asymptotically normal and therefore automatically asymptotically median-unbiased under $\mathsf{H}_1$. By Proposition \ref{prop:median}, $\hka^*$ is asymptotically median-unbiased under $\mathsf{H}_0$. Therefore, $\hka^*$ is always asymptotically median-unbiased. Then by Lemma 1 in \citet{kuchibhotla2024hulc}, given a generic number of fold $B$, we have the following property for the confidence interval $CI(B)$ as a function of $B$ constructed in Algorithm \ref{alg:CIh}:
  \[\lim_{n\to \infty} \Pr\left(\ka^* \in CI(B) \right) \ge 1 - 2^{1-B}.\]
  (i) $\CIh_{n,\alpha}$: By setting $B = \ceil{1 - \ln(\alpha) / \ln(2)}$, we have
  \[\lim_{n\to \infty} \Pr\left(\ka^* \in \CIh_{n,\alpha} \right) \ge 1 - 2^{1-B} \ge 1 - 2^{1 - (1 - \ln(\alpha) / \ln(2))} = 1 - \alpha.\]
  (ii) $\tCIh_{n,\alpha}$: By setting $B = \floor{1 - \ln(\alpha) / \ln(2)}$, we have
  \[\lim_{n\to \infty} \Pr\left(\ka^* \in \tCIh_{n,\alpha} \right) \ge 1 - 2^{1-B} \ge 1 - 2^{1 - (1 - \ln(\alpha) / \ln(2) - 1) } = 1 - 2\alpha.\]
  Hence, the proof is complete.
\end{proof}

\subsection{Theorem \ref{thm:asymp_pen}}
(i) $\mathsf{H}_1$: Under the alternative hypothesis $\mathsf{H}_1$, the proof is analogous to that of Proposition~\ref{thm:can}, thus omitted.

(ii) $\mathsf{H}_0$: Under the null hypothesis $\mathsf{H}_0$, the consistency proof is the same as in $\mathsf{H}_1$, thus omitted. For asymptotic distribution, we first define random variable 
\[Z_n \equiv n\dot{\mathbb{M}}_n^\dagger({L}^\dagger).\]
Note that $L^\dagger = \tilde{L}$, and 
\[\dot{\mathbb{M}}_n^\dagger({L}^\dagger) = \dot{\mathbb{M}}_n(\tilde{L}).\]
Recall that defined in Proposition~\ref{thm:can},
\[\dot{\mathbb{M}}_t 
= \ff{2\ka_t{\sigma}^{\ka, 2}_t\dot{\ka}_t - \ka^2_t\dot{\sigma}^{\ka, 2}_t}{{\sigma}^{\ka, 4}_t} , \quad
\dot{\mathbb{M}}_{n,t} = 
\ff{2\hat{\ka}_t\hat{\sigma}^{\ka, 2}_t\hat{\dot{\ka}}_t - \hat{\ka}^2_t\hat{\dot{\sigma}}^{\ka, 2}_t}{\hat{\sigma}^{\ka, 4}_t},\]
then, we have
\[\dot{\mathbb{M}}(\tilde{L}) = 0, \quad \dot{\mathbb{M}}_n^\dagger({L}^\dagger) = \dot{\mathbb{M}}_{n,\tilde{L}} = 
\ff{2\hat{\ka}_{\tilde{L}}\hat{\sigma}^{\ka, 2}_{\tilde{L}}\hat{\dot{\ka}}_{\tilde{L}} - \hat{\ka}^2_{\tilde{L}}\hat{\dot{\sigma}}^{\ka, 2}_{\tilde{L}}}{\hat{\sigma}^{\ka, 4}_{\tilde{L}}}.\]
Therefore,
\begin{equation*}
    \begin{aligned}
    Z_n &= n\ff{2\hat{\ka}_{\tilde{L}}\hat{\dot{\ka}}_{\tilde{L}}}{\hat{\sigma}^{\ka, 2}_{\tilde{L}}} - n\ff{\hat{\ka}^2_{\tilde{L}}\hat{\dot{\sigma}}^{\ka, 2}_{\tilde{L}}}{\hat{\sigma}^{\ka, 4}_{\tilde{L}}}.
    \end{aligned}
\end{equation*}
Since $\sqrt{n}\hat{\ka}_{\tilde{L}} \dto \mathcal{N}(0,\sigma^{\ka,2}_{\tilde{L}})$is $O_p(1)$, and $\sqrt{n} \hat{\dot{\ka}}_{\tilde{L}}$ also converges weakly to a Normal distribution and thus is $O_p(1)$. By Continuous Mapping Theorem, $Z_n = O_p(1)$. Note that the second derivative of ${\mathbb{M}}^\dagger(t)$ at ${L}^\dagger$ is $-2c$. Therefore, we invoke Theorem 3.2.16 in \citet{vandervaart1996weak} to conclude that 
\[n(\hL^\dagger - L^\dagger) = \ff{Z_n}{2c} + o_p(1).\]
Therefore, under $\mathsf{H}_0$, $|\hL^\dagger - L^\dagger| = O_p(n^{-1})$.
Then, we have
  \begin{equation*}
    \begin{aligned}
    \sqrt{n}(\hka^\dagger - \ka^\dagger) &= \sqrt{n}(\hka(\hL^\dagger) - \ka(L^\dagger))\\
    &= \sqrt{n}(\hka(L^\dagger) - \ka(L^\dagger)) + \sqrt{n}\hat{\dot{\ka}}(L^\dagger)(\hL^\dagger - L^\dagger) + o_p(\sqrt{n}|\hL^\dagger - L^\dagger|)\\
    &= \sqrt{n}(\hka(L^\dagger) - \ka(L^\dagger)) + O_p(n^{-1}) + o_p(n^{-1/2})\\
    & \dto \mathcal{N}(0, \sigma^{\ka,2}_{\tilde{L}})
    \end{aligned}
  \end{equation*}

\subsection{Proposition \ref{prop:AL}}
Given i.i.d. data $\{Z_i\}_{i=1}^n$, an estimator $\hat\tau$ of $\tau$ is said to be asymptotically linear if there exists a mean-zero square-integrable function $\varphi^{\tau}(Z)$ such that
\[\sqrt{n}(\hat\tau - \tau) = \ff{1}{\sqrt{n}}\sum_{i = 1}^n \varphi^{\tau}(Z_i) + o_p(1).\]
\begin{proposition}[Asymptotic Linearity]
  Under Assumptions \ref{assum:setting}, \ref{assum:smooth}, \ref{assum:unique_pen}, and \ref{assum:nonsin_pen}, $\hka^\dagger$ is an asymptotically linear estimator of $\ka^\dagger$ under both $\mathsf{H}_0, \mathsf{H}_1$, and $\hL^\dagger$ is an asymptotically linear estimator of $L^\dagger$ under $\mathsf{H}_1$.
  \label{prop:AL}
\end{proposition}

\begin{proof}
  (I) We first prove the asymptotic linearity of $\hL^\dagger$ and $\hka^\dagger$ under $\mathsf{H}_1$.

  For a generic parameter $\phi$ which admits an influence function, denote its influence function by $\IF{\phi;Z}$ (or $\IF{\phi}$ for brevity), where $Z$ is the observable random variables. Only in this proof, for ease of notation, we define $\dv_t \equiv \sigma^2(\hka_t)$, the asymptotic variance of $\hka_t$ at $t$, and $\ddv_t$ its derivative. Denote $\hdv_t, \hddv_t$ as the sample-analog estimates. See Lemma~\ref{lem:var} for the detailed expression of the components of $\dv_t$ and $\ddv_t$.
  Note that the survival curve of arm $A = a$, $S^a_t$, has the influence function $\IF{S^a_t}$. See \citet{westling2023inference} for details. Then, the influence function of $\ka_t$ exists with $\IF{\ka_t} = \int_0^t\IF{S^1_u} - \IF{S^0_u} \dy{u}.$ 

  By the first-order condition, we have 
  \[\partial_L \mathbb{M}_n^\dagger(L) = 0 \text{ at } \hL^\dagger,\]
  which implies
  \[g_n(\hL^\dagger) = 0,\]
  where 
  \[g_n(t) \equiv 2\hka_t\hnddka_t\hdv_t - \hka_t^2\hddv_t - c\hdv_t^2(2t - 2\tilde{L}).\]
  Define its probability limit as $g(t)$ where $g(\dL) = 0$.
  By linearization, we have 
  \[\sqrt{n}(\hdL - \dL) = -\ff{\sqrt{n}g_n(\dL)}{\dot{g}_n(\dL)} + o_p(1) = -\ff{\sqrt{n}g_n(\dL)}{\dot{g}(\dL)} + o_p(1)\]
  Since in the nonparametric statistical model, \[\sqrt{n}(\hka_t - \ka_t) = \ff{1}{\sqrt{n}}\sum_{i = 1}^n \IF{\ka_t;Z_i} + o_p(1), \sqrt{n}(\hnddka_t - \dka_t) = \ff{1}{\sqrt{n}}\sum_{i = 1}^n \left(\IF{S^1_t;Z_i} - \IF{S^0_t;Z_i}\right) + o_p(1),\] it only remains to show the asymptotic linearity of $\hdv_t$ and $\hddv_t$. 
  
  In the following, we prove for $\hdv_t$ and the proof for $\hddv_t$ is analogous and thus omitted. Recall the notations defined for Lemma~\ref{lem:var}. Note that $\hdv_t = \ff{1}{\beta}\hat{V}^1_t + \ff{1}{1-\beta}\hat{V}^0_t$, where the superscript represents the treatment arm. Therefore, we can focus on the asymptotic linearity of $V^a_t$ for $a = 0, 1$. WLOG, we suppress the superscript $a$ in the following. Define
  $V_t \equiv \int_0^t (\theta_t - \theta_v)^2\dy{\sigma_v^2}.$ %= \int_0^t (\theta_t - \theta_v)^2\ff{\alpha_v}{y_v}\dy{v}\]
  It is obvious that $(\theta_t - \theta_v)^2$ is asymptotically linear by the the asymptotic linearity of $\hat{S}_t$. As for $\sigma_t^2 \equiv \int_0^t \ff{\dy{\Lambda_u}}{y_u} $, it is well known that $n\int_0^t\ff{\dy{N_s}}{Y_s^2}$ is a uniformly consistent estimator of it \citep{andersen1993statistical}. 
  
  Then, by the Hadamard differentiability of integration operator and semiparametric theory \citep{bickel1993efficient, andersen1993statistical, ichimura2022influence}, we have
  \begin{equation*}
    \begin{aligned}
    \IF{\sigma_t^2} &= \int_0^t \IF{{1}/{y_u}}\dy{\Lambda_u} - \int_0^t \IF{\Lambda_u} \dy{\ff{1}{y_u}} + \ff{\IF{\Lambda_t}}{y_t},
    \end{aligned}
  \end{equation*}
  where the existence of $\IF{\Lambda_t}$ is implied by $\Lambda_t = -\ln S_t$, and the existence of $\IF{1/y_u}$ is implied by $\IF{y_u;Z} = \I{X\ge u} - y_u$.
  Therefore, $n\int_0^t\ff{\dy{N_s}}{Y_s^2}$ is also asymptotically linear. Then, we have influence function for $V_t$ as
  \begin{equation*}
    \begin{aligned}
    \IF{V_t} 
    = \int_0^t \IF{(\theta_t - \theta_v)^2}\dy{\sigma_v^2} - \int_0^t \IF{\sigma_v^2}\dy{(\theta_t - \theta_v)^2}.
    \end{aligned}
  \end{equation*}
  Hence, $\hdL$ is asymptotically linear, and by $\IF{\ka^\dagger} = \left\{S^1(L^\dagger) - S^0(L^\dagger)\right\} \cdot\IF{L^\dagger} + \IF{\ka_L}\mid_{L = L^\dagger}$, $\hdka$ is also asymptotically linear. Thus, we have $\varphi^{L^\dagger}(Z), \varphi^{\dka}(Z)$ as $\IF{L^\dagger;Z}, \IF{\ka^\dagger;Z}$.

  (II) We next prove the asymptotic linearity of $\hka^\dagger$ under $\mathsf{H}_0$. 

  By Theorem \ref{thm:asymp_pen}, $\sqrt{n}(\hka^\dagger - \ka^\dagger) = \sqrt{n}(\hka(\tilde{L}) - \ka(\tilde{L})) + o_p(1)$. As mentioned above, $\hka(t)$ is asymptotically linear for all $t$. Therefore, so is $\hka^\dagger$.
  
\end{proof}

\subsection{Proposition \ref{prop:bootstrap}}

\begin{proof}
  We show the proof for $\CIb_{\ka^\dagger, n, \alpha}$ under $\mathsf{H}_1$. The proofs for $\CIb_{\ka^\dagger, n, \alpha}$ under $\mathsf{H}_0$ and for $\CIb_{L^\dagger, n, \alpha}$ under $\mathsf{H}_1$ are analogous and omitted.\\
(I) First, the bootstrap consistency of $\CIb_{\ka^\dagger, n, \alpha}$ is implied by Proposition~\ref{prop:AL} and the Mammen's condition \citep{mammen2012does}. To see this, define $g_n(Z_i) \equiv \IF{\dka;Z_i}, \bar{g}_n = n^{-1}\sum_{i = 1}^n g_n(Z_i)$, and $\sigma^{\dka,2} = \E[\IF{\dka;Z_i}^2]$, $\sigma_n \equiv n^{-1/2}\sigma^{\dka}, T_n \equiv \bar{g}_n/\sigma_n$. On the bootstrap sample $O^*_b = \{ Z_i^*\}_{i=1}^n$, define $\bar{g}^*_n = n^{-1}\sum_{i = 1}^n g_n(Z_i^*)$, $T^*_n \equiv (\bar{g}^*_n - \bar{g}_n)/\sigma_n$. Let $G_n(t;F_0) \equiv \Pr(T_n \le t), G_n(t;F_n) \equiv \Pr(T^*_n \le t|O_0)$.

Then, by Proposition~\ref{prop:AL}, we have
\[T_n = \sqrt{n}(\hdka - \dka)/\sigma^{\dka} + o_p(1) \dto \mathcal{N}(0, 1)\]
implying $ G_{\infty}(t;F_0) = \Phi(t)$, where $\Phi(\cdot)$ is the CDF of the standard normal distribution. By the Mammen's condition
\[ \sup_{t \in \mathbb{R}} |G_n(t;F_n) - G_{\infty}(t;F_0)| \pto 0.\]

(II) Next, we prove the coverage properties. Define $Q_n(\alpha;\hat{\theta}, F_n)$ as the $\alpha$-th quantile of the bootstrap distribution of a generic statistic $\hat{\theta}$ conditional on observed data. 
For an arbitrary $\epsilon > 0$, we have
\begin{equation*}
  \begin{aligned}
    &\{G_n(T_n;F_n) < \alpha/2\}\\
    &= \{G_\infty(T_n;F_0) < \alpha/2-  (G_n(T_n;F_n) - G_\infty(T_n;F_0))\}\\
    &\subseteq  \left(\{G_\infty(T_n;F_0) < \alpha/2 + \epsilon\} \cap \{|G_n(T_n;F_n) - G_\infty(T_n;F_0)| \le \epsilon\}\right)\cup \{|G_n(T_n;F_n) - G_\infty(T_n;F_0)| > \epsilon\}\\
    &\subseteq  \{G_\infty(T_n;F_0) < \alpha/2 + \epsilon\}\cup \{|G_n(T_n;F_n) - G_\infty(T_n;F_0)| > \epsilon\}\\
    &\{G_n(T_n;F_n) \le 1- \alpha/2\}\\
    &= \{G_\infty(T_n;F_0) \le 1 - \alpha/2-  (G_n(T_n;F_n) - G_\infty(T_n;F_0))\}\\
    &\supseteq  \{G_\infty(T_n;F_0) \le 1 - \alpha/2 - \epsilon\} \cap \{|G_n(T_n;F_n) - G_\infty(T_n;F_0)| \le \epsilon\}.
  \end{aligned}
\end{equation*}
Then, we have
\begin{equation*}
  \begin{aligned}
  &\lim_{n\to\infty}\Pr\left(T_n \in  \left[Q_n(\alpha/2;T_n, F_n), Q_n(1-\alpha/2;T_n, F_n)\right]\right) \\
  &= \lim_{n\to\infty}\Pr\left( \alpha/2 \le G_n(T_n;F_n) \le 1- \alpha/2 \right)\\
  &= \lim_{n\to\infty}\Pr\left( G_n(T_n;F_n) \le 1- \alpha/2 \right) - \Pr\left(  G_n(T_n;F_n) < \alpha/2 \right)\\
  &\ge \lim_{n\to\infty} \Pr\left(\{G_\infty(T_n;F_0) \le 1 - \alpha/2 - \epsilon\} \cap \{|G_n(T_n;F_n) - G_\infty(T_n;F_0)| \le \epsilon\}\right)\\
  & \quad - \Pr\left(  \{G_\infty(T_n;F_0) < \alpha/2 + \epsilon\}\cup \{|G_n(T_n;F_n) - G_\infty(T_n;F_0)| > \epsilon\}\right)\\
  &\ge \lim_{n\to\infty} \Pr\left(\{G_\infty(T_n;F_0) \ge 1 - \alpha/2 - \epsilon\}\right) - \Pr\left(\{|G_n(T_n;F_n) - G_\infty(T_n;F_0)| > \epsilon\}\right)  \\
  & \quad - \Pr\left(  \{G_\infty(T_n;F_0) < \alpha/2 + \epsilon\}\right)- \Pr\left(  \{|G_n(T_n;F_n) - G_\infty(T_n;F_0)| > \epsilon\}\right)\\
  &= \lim_{n\to\infty} \Pr\left(\{G_\infty(T_n;F_0) \ge 1 - \alpha/2 - \epsilon\}\right)  - \Pr\left(  \{G_\infty(T_n;F_0) < \alpha/2 + \epsilon\}\right)\\
  &= \Pr\left(\{G_\infty(T_\infty;F_0) \ge 1 - \alpha/2 - \epsilon\}\right)  - \Pr\left(  \{G_\infty(T_\infty;F_0) < \alpha/2 + \epsilon\}\right)\\
  &= 1 - \alpha - 2\epsilon.
  \end{aligned}
\end{equation*}
Since $\epsilon$ is arbitrary, we have 
\[\lim_{n\to\infty}\Pr\left(T_n \in  \left[Q_n(\alpha/2;T_n, F_n), Q_n(1-\alpha/2;T_n, F_n)\right]\right) \ge 1-\alpha.\]
Define $q_{\alpha}(\cdot)$ as the $\alpha$-th quantile of a vector. Therefore, we have
\begin{equation*}
  \begin{aligned}
    \lim_{n\to\infty}\Pr\left(\ka^\dagger \in \CIb_{\ka^\dagger, n, \alpha}\right) 
    &=\lim_{n\to\infty} \Pr\left(\ka^\dagger \in \left[ q_{\alpha/2}(\{\hka^\dagger_b\}_{b = 1}^B), q_{1-\alpha/2}(\{\hka^\dagger_b\}_{b = 1}^B) \right]\right) \\
    & \text{(Choosing large B)}\\
    &\approx \lim_{n\to\infty}\Pr\left(\ka^\dagger \in \left[ Q_n(\alpha/2;\hka^\dagger, F_n), Q_n(1-\alpha/2;\hka^\dagger, F_n) \right]\right)\\
    &= \lim_{n\to\infty}\Pr\left(\ka^\dagger - \hdka \in \left[ Q_n(\alpha/2;\hka^\dagger, F_n)- \hdka, Q_n(1-\alpha/2;\hka^\dagger, F_n)- \hdka \right]\right)\\
    &= \lim_{n\to\infty}\Pr\left(\ka^\dagger - \hdka \in \left[ Q_n(\alpha/2;\hka^\dagger - \dka, F_n), Q_n(1-\alpha/2;\hka^\dagger - \dka, F_n) \right]\right)\\
    &= \lim_{n\to\infty}\Pr\left(\hdka - \dka \in \left[-Q_n(1-\alpha/2;\hka^\dagger - \dka, F_n), -Q_n(\alpha/2;\hka^\dagger - \dka, F_n)  \right]\right)\\
    & \quad \text{(By the asymptotic symmetry of the bootstrap distribution of $\hka^\dagger - \dka$.)}\\
    &= \lim_{n\to\infty}\Pr\left(\hdka - \dka \in \left[Q_n(\alpha/2;\hka^\dagger - \dka, F_n), Q_n(1-\alpha/2;\hka^\dagger - \dka, F_n)  \right]\right)\\
    & \quad \text{(By the monotonicity of scale-location transform.)}\\
    &= \lim_{n\to\infty}\Pr\left(T_n \in \left[Q_n(\alpha/2;T_n, F_n), Q_n(1-\alpha/2;T_n, F_n)  \right]\right)\\
    &\ge 1 - \alpha.
  \end{aligned}
\end{equation*}
Hence, the proof is complete.
\end{proof}

\subsection{Theorem \ref{thm:dt}}
\begin{proof}
  (1) $\hL^\dagger$:\\
  Define arbitrary sequence of integers $a_n$ that goes to $\infty$ as $n \to \infty$. Given any positive constant $C$, We have
  \begin{equation*}
    \begin{aligned}
      \Pr(a_n|\hL^\dagger - L^\dagger| > C) 
      &= \Pr(|\hL^\dagger - L^\dagger| > C/a_n)\\
      &\le \sum_{j \ne j^\dagger} \Pr(\hat{j}^\dagger = j)\\
      &= \sum_{j \ne j^\dagger} \Pr(\mathbb{M}_{nj}^\dagger - \mathbb{M}_{ni}^\dagger > 0, \forall i \ne j)\\
      &= \sum_{j \ne j^\dagger} \Pr(\sqrt{n}(\mathbb{M}_{nj}^\dagger - \mathbb{M}_{ni}^\dagger - (\mathbb{M}_{j}^\dagger - \mathbb{M}_{i}^\dagger)) > -\sqrt{n}(\mathbb{M}_{j}^\dagger - \mathbb{M}_{i}^\dagger), \forall i \ne j)\\
      &\le \sum_{j \ne j^\dagger} \Pr(\sqrt{n}(\mathbb{M}_{nj}^\dagger - \mathbb{M}_{nj^\dagger}^\dagger - (\mathbb{M}_{j}^\dagger - \mathbb{M}_{j^\dagger}^\dagger)) > -\sqrt{n}(\mathbb{M}_{j}^\dagger - \mathbb{M}_{j^\dagger}^\dagger))
    \end{aligned}
  \end{equation*}
Since $\sqrt{n}(\mathbb{M}_{nj}^\dagger - \mathbb{M}_{nj^\dagger}^\dagger - (\mathbb{M}_{j}^\dagger - \mathbb{M}_{j^\dagger}^\dagger))$ is $O_p(1)$ or $o_p(1)$, and by Assumption~\ref{assum:unique_dt}, $-\sqrt{n}(\mathbb{M}_{j}^\dagger - \mathbb{M}_{j^\dagger}^\dagger) \to \infty$. Therefore, we have $\Pr(a_n|\hL^\dagger - L^\dagger| > C) \to 0$. This proves the consistency of $\hL^\dagger$ under arbitrary rate $a_n$. Specifically, we have $(\hL^\dagger - L^\dagger) = o_p(n^{-1/2})$ and use this below.

(2) $\hka^\dagger$:\\
\begin{equation*}
  \begin{aligned}
    \hka(\hL^\dagger) &= \hka(L^\dagger) + \hka'(L^\dagger)(\hL^\dagger - L^\dagger) + o_p(|\hL^\dagger - L^\dagger|)\\
    &= \hka(L^\dagger) + \hka'(L^\dagger)(\hL^\dagger - L^\dagger) + o_p(n^{-1/2})\\
    \sqrt{n}(\hka(\hL^\dagger) - \ka(L^\dagger)) &= \sqrt{n}(\hka(L^\dagger) - \ka(L^\dagger)) +\hka'(L^\dagger)o_p(1) + o_p(1)\\
    &= \sqrt{n}(\hka(L^\dagger) - \ka(L^\dagger))  + o_p(1)\\
  \end{aligned}
\end{equation*}

(3) $\hat\sigma_{\hat{L}^\dagger}^{\ka, 2}$:\\
\begin{equation*}
  \begin{aligned}
    \hat\sigma_{\hat{L}^\dagger}^{\ka, 2} &= \hat\sigma_{{L}^\dagger}^{\ka, 2} + \hat{\dot{\sigma}}_{{L}^\dagger}^{\ka, 2}(\hL^\dagger - L^\dagger) + o_p(|\hL^\dagger - L^\dagger|) \pto \sigma_{{L}^\dagger}^{\ka, 2} \text{ as } n \to \infty.
  \end{aligned}
\end{equation*}

\end{proof}

\subsection{Proposition \ref{prop:grid}}

\begin{proof}
  (1) By Assumption \ref{assum:unique_pen} and \ref{assum:unique_dt}, $L^\dagger$ and $L^\dagger_{ct}$ are well-defined. Then, $|L^\dagger - L^\dagger_{ct}| \le \bar{k}/(m_n-1) = O(n^{-\gamma}) \to 0$. Then, by the continuity of $\ka(\cdot)$, we have $\ka^\dagger \to \ka^\dagger_{ct}$.

  (2) In the following, we discuss two cases.

  (i) Under $\mathsf{H}_1$:
  Similar to the proof of Theorem \ref{thm:dt},
  \begin{equation*}
    \begin{aligned}
      \Pr(a_n|\hL^\dagger - L^\dagger| > C) 
      &\le \sum_{j \ne j^\dagger} \Pr(\sqrt{n}(\mathbb{M}_{nj}^\dagger - \mathbb{M}_{nj^\dagger}^\dagger - (\mathbb{M}_{j}^\dagger - \mathbb{M}_{j^\dagger}^\dagger)) > -\sqrt{n}(\mathbb{M}_{j}^\dagger - \mathbb{M}_{j^\dagger}^\dagger))\\
      &\le \sum_{j \ne j^\dagger}\ff{c_1}{\sqrt{n}(\mathbb{M}_{j^\dagger}^\dagger - \mathbb{M}_{j}^\dagger)}\exp\left(-\ff{n(\mathbb{M}_{j^\dagger}^\dagger - \mathbb{M}_{j}^\dagger)^2}{2\sigma^2_{j,j^\dagger}}\right) + \ff{c_2}{\sqrt{n}} ,
    \end{aligned}
  \end{equation*}
  where $c_1,c_2$ are positive constants.
  The first term in the last step is by Mill's inequality for Gaussian random variables, where by the Equation \eqref{eq:asympMn} in the proof of Proposition~\ref{thm:can}, we have $\sqrt{n}(\mathbb{M}_{nj}^\dagger - \mathbb{M}_{ni}^\dagger - (\mathbb{M}_{j}^\dagger - \mathbb{M}_{i}^\dagger))  = \sqrt{n}(\mathbb{M}_{nj} - \mathbb{M}_{ni} - (\mathbb{M}_{j} - \mathbb{M}_{i})) \dto \mathcal{N}(0, \sigma^{2}_{j,i})$. The second term in the last step is by Berry-Esseen theorem.

  Note that (a) when $n$ is large, $\mathbb{M}_{j^\dagger} - \mathbb{M}_{j} \ge \mathbb{M}_{j^\dagger} - \max(\mathbb{M}_{j^\dagger+1}, \mathbb{M}_{j^\dagger-1})$; (b) $\sigma^2_{j,j^\dagger}$ is upper bounded; (c) by Talyor expansion of $\mathbb{M}^\dagger$ and the fact $\dot{\mathbb{M}}_{j^\dagger}^\dagger = O(n^{-\gamma})$, we have $-\sqrt{n}(\mathbb{M}_{j^\dagger+1}^\dagger - \mathbb{M}_{j^\dagger}^\dagger) = \sqrt{n}|\mathbb{M}_{j^\dagger+1}^\dagger - \mathbb{M}_{j^\dagger}^\dagger| = \sqrt{n}(|\dot{\mathbb{M}}_{j^\dagger}^\dagger||L_{j^\dagger+1} - L^\dagger| + O(|L_{j^\dagger+1} - L^\dagger|^2)) = O(n^{1/2 - 2\gamma}) \to \infty$ when $\gamma < 1/4$.
  
  Therefore, when $\gamma < 1/4$, we have
  \begin{equation*}
    \begin{aligned}
      \Pr(a_n|\hL^\dagger - L^\dagger| > C) 
      &\le O(n^\gamma)\ff{1}{\sqrt{n}O(n^{-2\gamma})}\exp\left(-O(n \cdot n^{-4\gamma})\right) +  O(n^{\gamma - 1/2})\\
      &= O(n^{3\gamma-1/2})\exp\left(-O(n^{1-4\gamma})\right) + O(n^{\gamma - 1/2})\to 0.
    \end{aligned}
  \end{equation*}
  Then, with a similar argument as in the proof of Theorem~\ref{thm:dt}, we finished the proof.

  (ii) Under $\mathsf{H}_0$:\\
  Note that $\mathbb{M}_j \equiv 0, \mathbb{M}_j^\dagger = -c(L_j - \tilde{L})^2, L^\dagger = \tilde{L}, \mathbb{M}_{j^\dagger}^\dagger = 0$, and $n\mathbb{M}_{nj} = \chi^2_1 + o_p(1)$ for all $j$ where $\chi^2_1$ is the chi-squared distribution with $1$ degree of freedom.
  \begin{equation*}
    \begin{aligned}
      \Pr(a_n|\hL^\dagger - L^\dagger| > C) 
      &\le \sum_{j \ne j^\dagger} \Pr(n(\mathbb{M}_{nj}^\dagger - \mathbb{M}_{nj^\dagger}^\dagger - (\mathbb{M}_{j}^\dagger - \mathbb{M}_{j^\dagger}^\dagger)) > -n(\mathbb{M}_{j}^\dagger - \mathbb{M}_{j^\dagger}^\dagger))\\
      &= \sum_{j \ne j^\dagger} \Pr\left\{n(\mathbb{M}_{nj}-\mathbb{M}_{j}) - n(\mathbb{M}_{nj^\dagger} -  \mathbb{M}_{j^\dagger}) > -n(\mathbb{M}_{j}^\dagger - \mathbb{M}_{j^\dagger}^\dagger)\right\}\\
      &= \sum_{j \ne j^\dagger} \Pr\left\{n\mathbb{M}_{nj} - n\mathbb{M}_{nj^\dagger} > nc(L_j - \tilde{L})^2\right\}\\
      & \le \sum_{j \ne j^\dagger} \Pr\left\{n\mathbb{M}_{nj} - n\mathbb{M}_{nj^\dagger} > \ff{c}{k^2}n^{1-2\gamma}\right\}\\
      & \le \sum_{j \ne j^\dagger} \Pr\left\{n\mathbb{M}_{nj} > \ff{c}{k^2}n^{1-2\gamma}\right\}\\
      & (\text{by the tail bound of $\chi^2_k$ distribution \citep{laurent2000adaptive}})\\
      & \lesssim O(n^\gamma)\exp\left\{ -O(n^{1-2\gamma}) \right\} + O(n^{\gamma-1/2})
    \end{aligned}
  \end{equation*}
  Therefore, under $\mathsf{H}_0$, any $\gamma < 1/2$ ensures $\Pr(a_n|\hL^\dagger - L^\dagger| > C) \to 0$ for arbitrary $a_n$.

  (iii) Together with results under $\mathsf{H}_1$, we conclude that any $\gamma < 1/4$ ensures $\Pr(a_n|\hL^\dagger - L^\dagger| > C) \to 0$ for arbitrary $a_n$.

  (3) $|\hat{L}^\dagger - L^\dagger_{ct}| \le |\hat{L}^\dagger - L^\dagger| + |L^\dagger - L_{ct}^\dagger| = |\hat{L}^\dagger - L^\dagger| + O(n^{-\gamma})$. Therefore, when $\gamma < 1/4$, 
  \[|\hat{L}^\dagger - L^\dagger_{ct}| = o_p(n^{-1/2}) + O(n^{-\gamma}) = O_p(n^{-\gamma}) \pto 0.\]

\end{proof}

\section{Additional Simulation Details and Results}

For easy reference, we iterate the nine scenarios designed to capture common patterns observed in randomized trials: (1) \texttt{null}: No treatment effect (null hypothesis). (2) \texttt{ph}: Proportional hazards. (3) \texttt{early}: Early treatment effects, where the benefit diminishes over time. (4) \texttt{tran}: Transient treatment effects, where the benefit appears temporarily and then disappears. (5) \texttt{cs}: Crossing survival curves. (6) \texttt{msep}: Middle separation. (7) \texttt{delay\_1}: Moderately delayed treatment effects, with a gradual onset of benefit. (8) \texttt{delay\_2}: Heavily delayed treatment effects, where the benefit takes a longer time to emerge. (9) \texttt{delaycon}: Delayed treatment effects with converging survival curves, modeling a scenario where the benefit emerges late and then diminishes.

\subsection{Data Generating Details of Each Scenario}

We provide the detailed data generating process of the treatment arm for each scenario in the simulation study. In particular, we list the parameter values of the piecewise exponential distribution for each scenario in Table~\ref{tab:scenario_parameters} below. A piecewise exponential distribution is a generalization of the exponential distribution where the hazard rate (or rate parameter) is allowed to change across different time intervals. This flexibility enables it to model survival data or event times where the risk of an event is not constant over time. The parameter \texttt{Change Point} $(t_0= 0, t_1, \ldots, t_k, t_{k+1} = \infty)$ refers to the time points where the hazard rate change, dividing the time axis into $[0, t_1), [t_1, t_2), \ldots, [t_k, \infty)$. The parameter \texttt{Rate} $\lambda_0, \lambda_1, \ldots, \lambda_k$ refers to the constant hazard rate of the exponential distribution in each time interval: the hazard function $\lambda(u) = \lambda_i$ if $u \in [t_i, t_{i+1})$. The survival function is $S(t) = \exp\left(-\int_0^t \lambda(u) \dy{u}\right)$. Note that the control arm event time is generated by an exponential distribution with a constant hazard rate of 1. Figures \ref{fig:scenario_null} to \ref{fig:scenario_delaycon} illustrate the survival curves (top left), hazard functions (bottom left), criterion functions $\mathbb{M}$ and the penalized versions under different penalty strengths $c$ (top right), where $\tilde{L} = 2.2$ is the midpoint, and treatment effect $\ka$ (RMST\_diff) compared to the square root of $\mathbb{M}$ (sqrt(M)) (bottom right).

\begin{table}
  \centering
  \caption{Rate and Change-Point for Treatment Arm in Each Scenario}
  \begin{tabular}{lll}
  \toprule
  \textbf{Scenario} & \textbf{Rate}         & \textbf{Change Point}         \\ 
  \midrule
  \texttt{null}     & (1)                      & (0)                      \\ 
  \texttt{ph}       & (0.75)                   & (0)                      \\ 
  \texttt{early}    & (0.65, 1)                & (0, 0.5)                 \\ 
  \texttt{tran}     & (0.5, 1.5, 1)            & (0, 0.6, 1.2)            \\ 
  \texttt{cs}       & (0.5, 1.4)               & (0, 0.5)                 \\ 
  \texttt{msep}     & (0.5, 1.1)               & (0, 0.5)                 \\ 
  \texttt{delay\_1} & (1, 0.7)                 & (0, 0.2)                 \\ 
  \texttt{delay\_2} & (1, 0.7)                 & (0, 0.4)                 \\ 
  \texttt{delaycon} & (1, 0.7, 1.2)            & (0, 0.2, 1)              \\ 
  \bottomrule
  \end{tabular}
  \label{tab:scenario_parameters}
\end{table}

\begin{figure}
  \centering
  \includegraphics[width=0.8\textwidth]{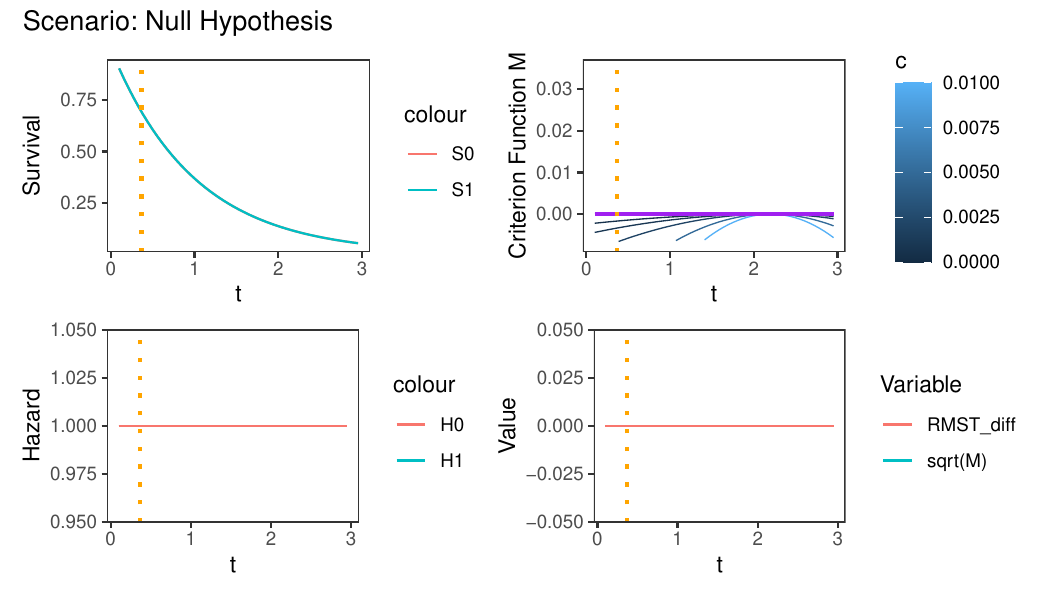}
  \caption{DGP of Scenario \texttt{null}. The orange dotted line is the optimal restriction time $L^*$, as defined in Equation \eqref{eq:estimand}. In this scenario, $L^*$ is not unique.}
  \label{fig:scenario_null}
\end{figure}

\begin{figure}
  \centering
  \includegraphics[width=0.8\textwidth]{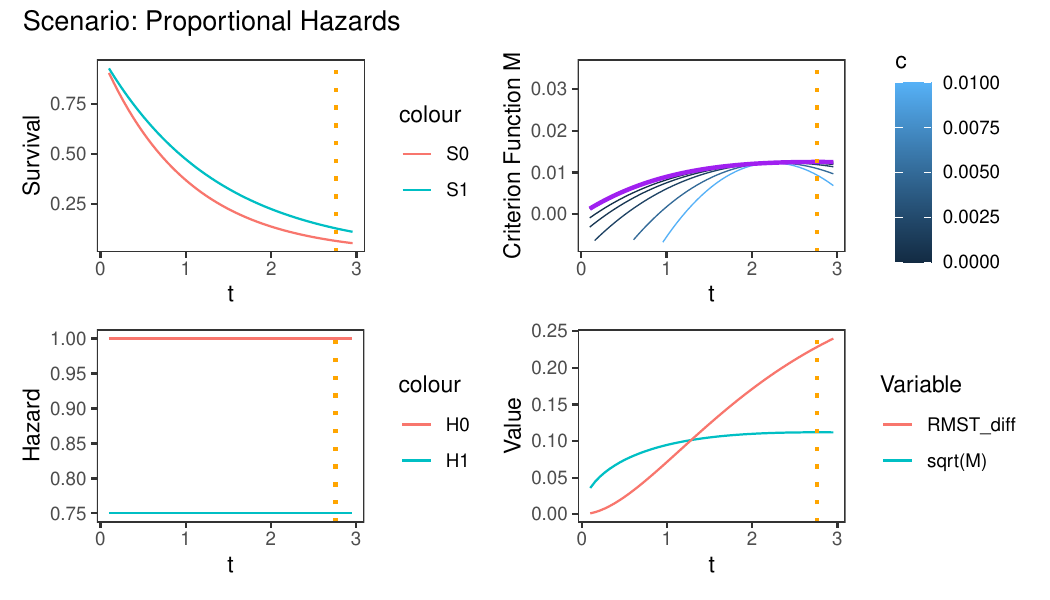}
  \caption{DGP of Scenario \texttt{ph}. The orange dotted line is the optimal restriction time $L^*$, as defined in Equation \eqref{eq:estimand}.}
  \label{fig:scenario_ph}
\end{figure}

\begin{figure}
  \centering
  \includegraphics[width=0.8\textwidth]{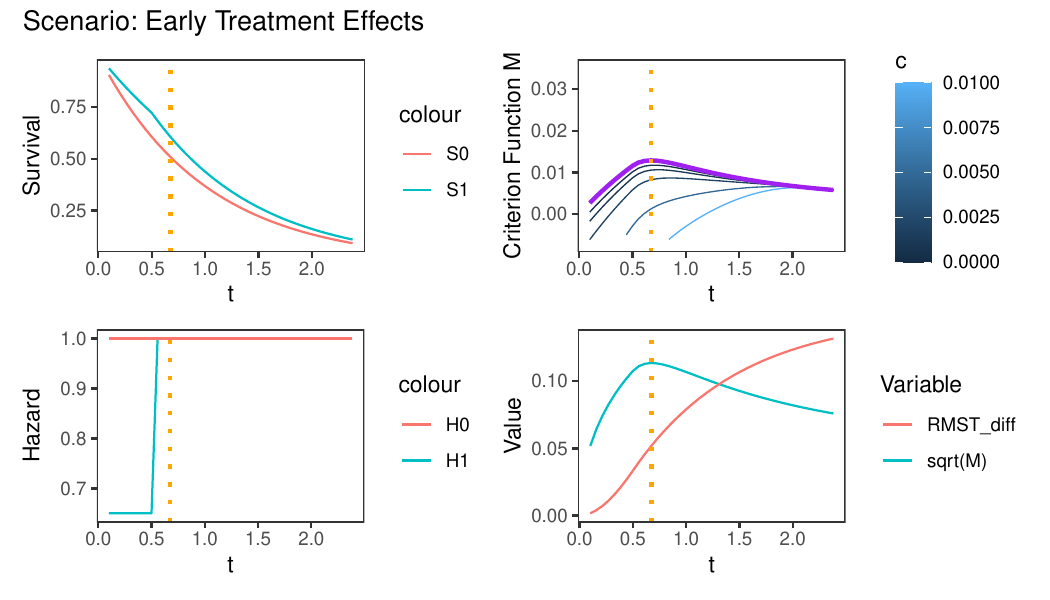}
  \caption{DGP of Scenario \texttt{early}. The orange dotted line is the optimal restriction time $L^*$, as defined in Equation \eqref{eq:estimand}.}
  \label{fig:scenario_early}
\end{figure}

\begin{figure}
  \centering
  \includegraphics[width=0.8\textwidth]{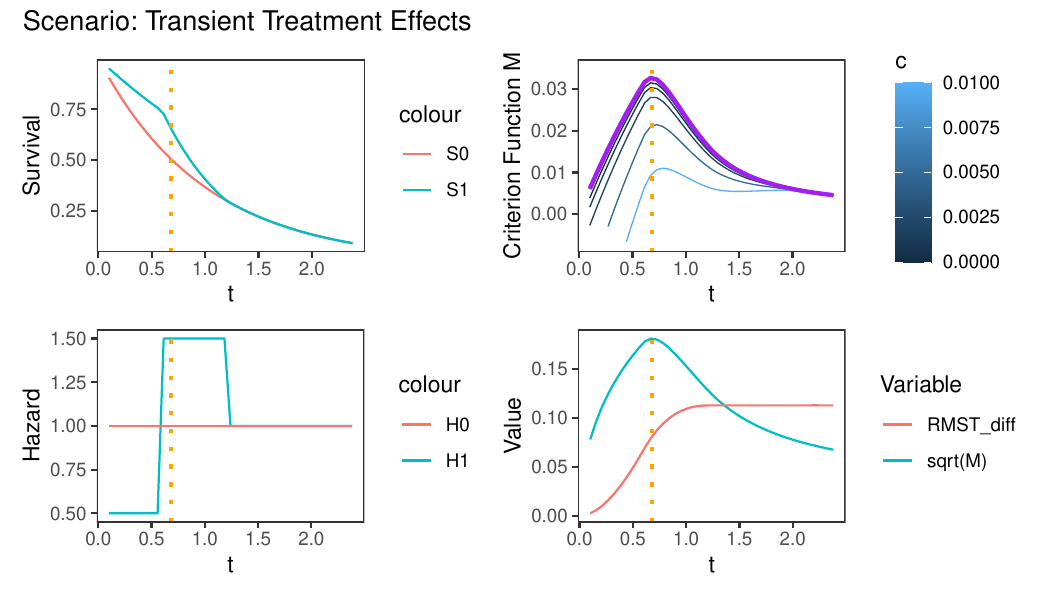}
  \caption{DGP of Scenario \texttt{tran}. The orange dotted line is the optimal restriction time $L^*$, as defined in Equation \eqref{eq:estimand}.}
  \label{fig:scenario_tran}
\end{figure}

\begin{figure}
  \centering
  \includegraphics[width=0.8\textwidth]{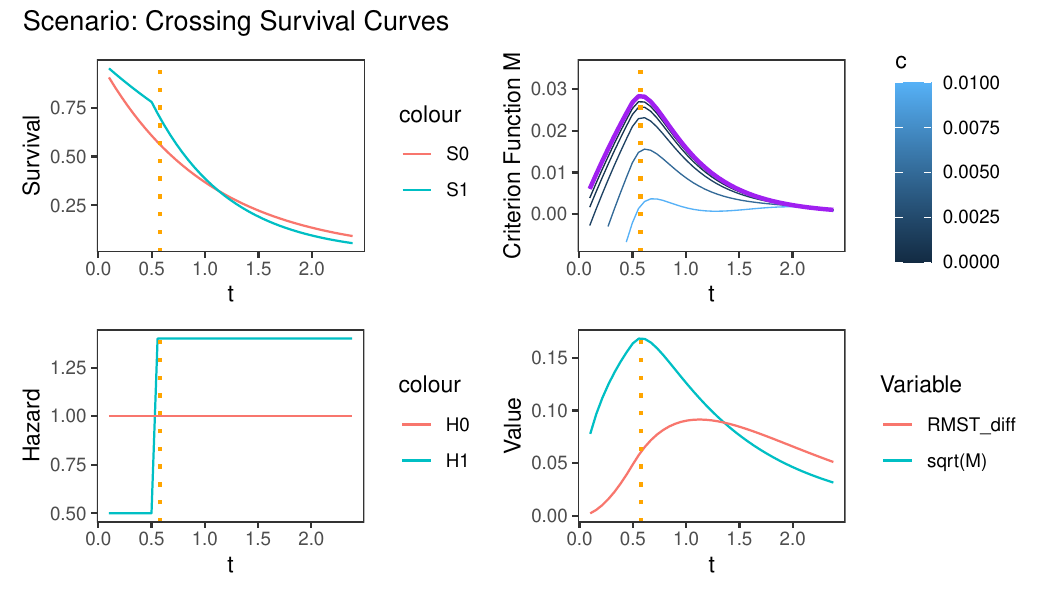}
  \caption{DGP of Scenario \texttt{cs}. The orange dotted line is the optimal restriction time $L^*$, as defined in Equation \eqref{eq:estimand}.}
  \label{fig:scenario_cs}
\end{figure}

\begin{figure}
  \centering
  \includegraphics[width=0.8\textwidth]{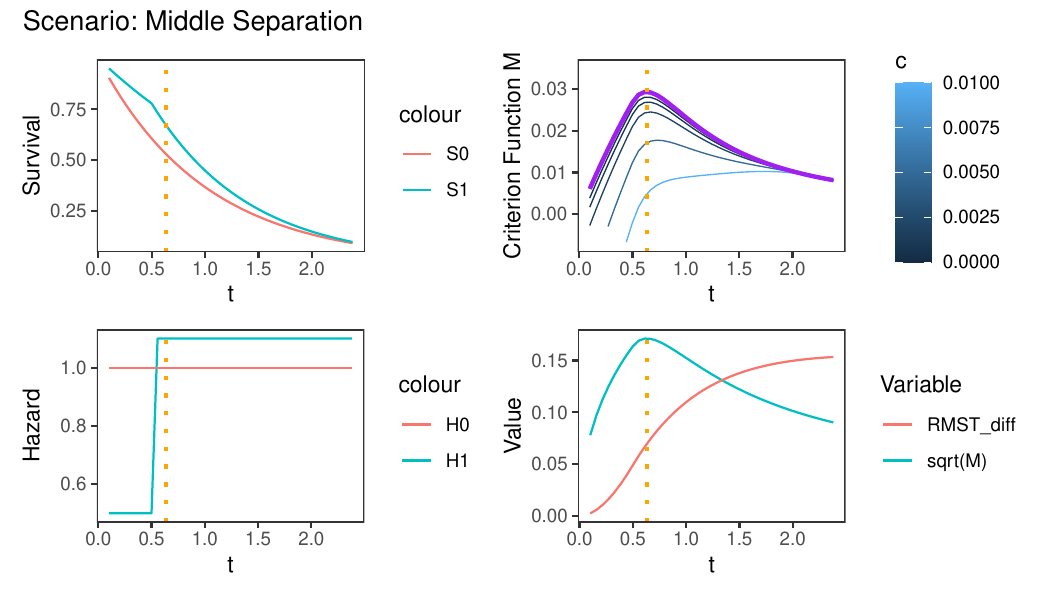}
  \caption{DGP of Scenario \texttt{msep}. The orange dotted line is the optimal restriction time $L^*$, as defined in Equation \eqref{eq:estimand}.}
  \label{fig:scenario_msep}
\end{figure}

\begin{figure}
  \centering
  \includegraphics[width=0.8\textwidth]{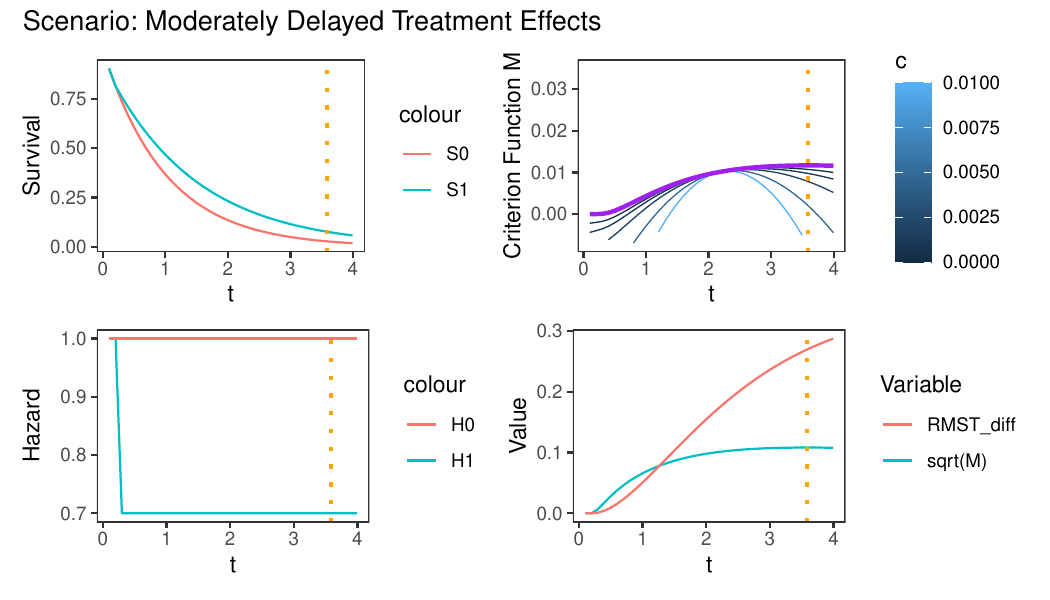}
  \caption{DGP of Scenario \texttt{delay\_1}. The orange dotted line is the optimal restriction time $L^*$, as defined in Equation \eqref{eq:estimand}.}
  \label{fig:scenario_delay_1}
\end{figure}

\begin{figure}
  \centering
  \includegraphics[width=0.8\textwidth]{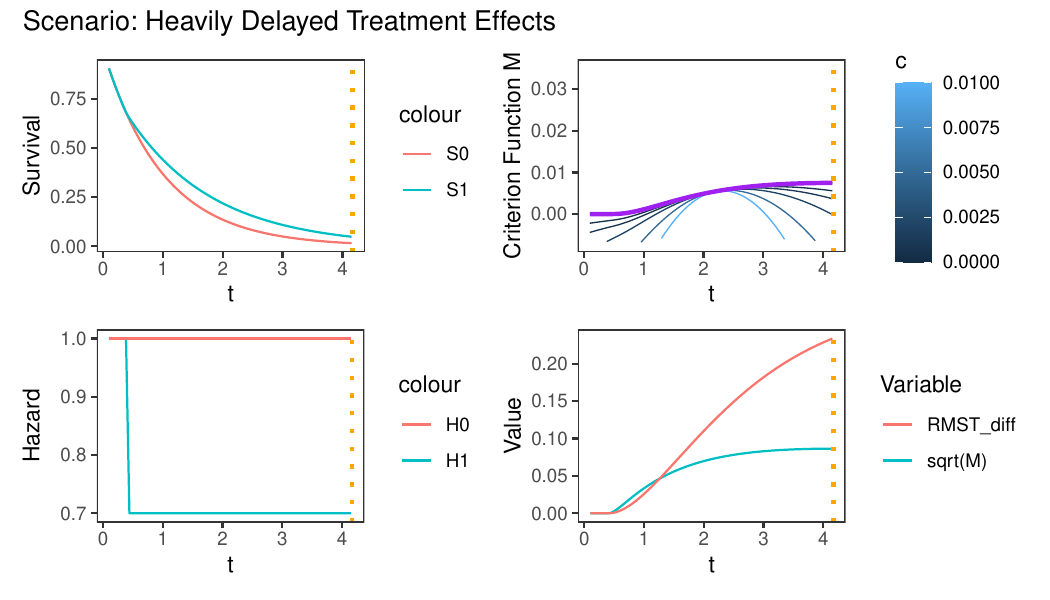}
  \caption{DGP of Scenario \texttt{delay\_2}. The orange dotted line is the optimal restriction time $L^*$, as defined in Equation \eqref{eq:estimand}.}
  \label{fig:scenario_delay_2}
\end{figure}

\begin{figure}
  \centering
  \includegraphics[width=0.8\textwidth]{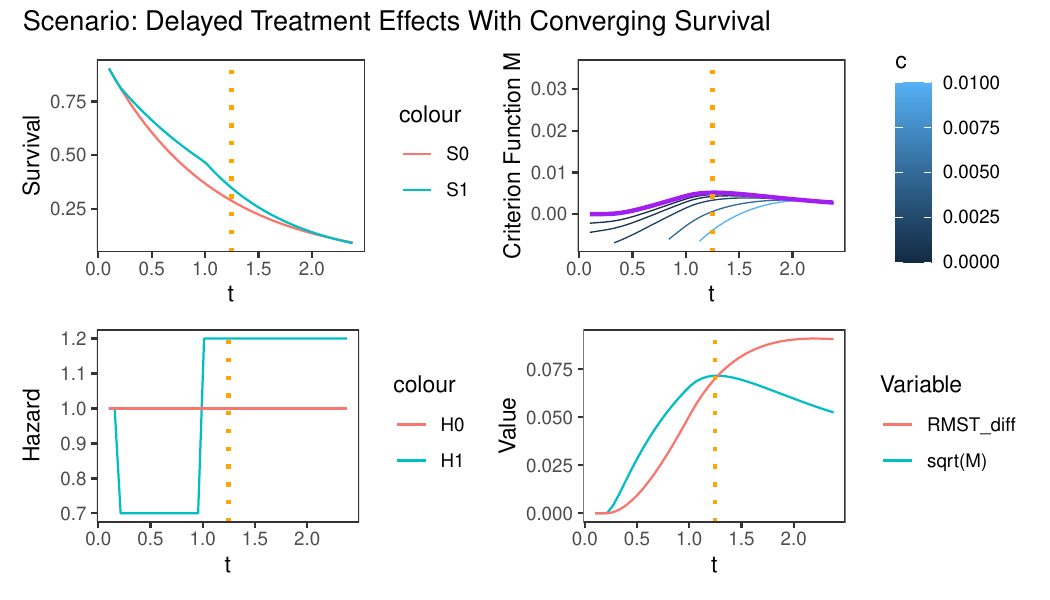}
  \caption{DGP of Scenario \texttt{delaycon}. The orange dotted line is the optimal restriction time $L^*$, as defined in Equation \eqref{eq:estimand}.}
  \label{fig:scenario_delaycon}
\end{figure}

\subsection{Additional Simulation Results: Estimation of $\Lopt$}

The simulation results for the performance of optimal restriction time estimators are presented in Tables \ref{tab:L_n300}, \ref{tab:L_n600}, and \ref{tab:L_n1000} for sample sizes $n = 300, 600, 1000$, respectively. Three metrics are evaluated across 2000 repetitions: (1) Bias, representing the empirical bias; (2) SD, the empirical standard deviation; and (3) RMSE, the root mean square error.

\begin{table}[ht]
  \centering
  \caption{Simulation results for optimal restriction time estimators for sample size $n = 300$.}
  \begin{tabular}{l|rrr|rrr|rrr}
  \toprule
  Scenario & \multicolumn{3}{c|}{AdaRMST.ct}   &  \multicolumn{3}{c|}{AdaRMST.dt}  & \multicolumn{3}{c}{hulc}  \\
  & Bias & SD & RMSE & Bias & SD & RMSE & Bias & SD & RMSE \\
  \midrule
  null & -0.230 & 0.689 & 0.726 & -0.080 & 0.387 & 0.395 & -0.637 & 1.318 & 1.464 \\
  ph & -0.204 & 0.827 & 0.851 & 0.005 & 0.483 & 0.483 & -0.568 & 1.323 & 1.440 \\
  early & 0.472 & 0.803 & 0.931 & -0.065 & 0.602 & 0.605 & 0.414 & 0.971 & 1.055 \\
  tran & 0.069 & 0.384 & 0.390 & 0.218 & 0.443 & 0.494 & 0.014 & 0.418 & 0.418 \\
  cs & 0.141 & 0.539 & 0.557 & 0.238 & 0.523 & 0.574 & 0.115 & 0.552 & 0.563 \\
  msep & 0.209 & 0.534 & 0.573 & 0.347 & 0.502 & 0.610 & 0.146 & 0.537 & 0.557 \\
  delay\_1 & -0.215 & 0.785 & 0.814 & -0.283 & 0.412 & 0.500 & -1.030 & 1.310 & 1.666 \\
  delay\_2 & -0.288 & 0.804 & 0.854 & -0.315 & 0.448 & 0.547 & -1.753 & 1.439 & 2.267 \\
  delaycon & 0.132 & 0.658 & 0.671 & 0.215 & 0.393 & 0.448 & 0.159 & 1.031 & 1.043 \\
  \bottomrule
  \end{tabular}
  \label{tab:L_n300}
\end{table}

\begin{table}[ht]
  \caption{Simulation results for optimal restriction time estimators for sample size $n = 600$.}
  \centering
  \begin{tabular}{l|rrr|rrr|rrr}
  \toprule
  Scenario & \multicolumn{3}{c|}{AdaRMST.ct}   &  \multicolumn{3}{c|}{AdaRMST.dt}  & \multicolumn{3}{c}{hulc}  \\
  & Bias & SD & RMSE & Bias & SD & RMSE & Bias & SD & RMSE \\
  \midrule
  null & -0.069 & 0.424 & 0.430 & -0.013 & 0.182 & 0.182 & -0.534 & 1.376 & 1.476 \\
  ph & -0.102 & 0.622 & 0.630 & 0.089 & 0.345 & 0.357 & -0.288 & 1.222 & 1.255 \\
  early & 0.401 & 0.662 & 0.773 & -0.080 & 0.539 & 0.545 & 0.302 & 0.765 & 0.822 \\
  tran & -0.005 & 0.149 & 0.149 & 0.086 & 0.254 & 0.268 & -0.031 & 0.154 & 0.157 \\
  cs & 0.018 & 0.248 & 0.249 & 0.092 & 0.328 & 0.341 & 0.023 & 0.199 & 0.201 \\
  msep & 0.096 & 0.318 & 0.332 & 0.237 & 0.385 & 0.452 & 0.058 & 0.263 & 0.269 \\
  delay\_1 & -0.055 & 0.521 & 0.523 & -0.214 & 0.323 & 0.388 & -0.568 & 1.066 & 1.207 \\
  delay\_2 & -0.105 & 0.546 & 0.556 & -0.238 & 0.305 & 0.387 & -1.218 & 1.257 & 1.750 \\
  delaycon & 0.130 & 0.472 & 0.489 & 0.230 & 0.289 & 0.369 & 0.087 & 0.854 & 0.858 \\
  \bottomrule
  \end{tabular}
  \label{tab:L_n600}
\end{table}

\begin{table}[ht]
  \caption{Simulation results for optimal restriction time estimators for sample size $n = 1000$.}
  \centering
  \begin{tabular}{l|rrr|rrr|rrr}
  \toprule
  Scenario & \multicolumn{3}{c|}{AdaRMST.ct}   &  \multicolumn{3}{c|}{AdaRMST.dt}  & \multicolumn{3}{c}{hulc}  \\
  & Bias & SD & RMSE & Bias & SD & RMSE & Bias & SD & RMSE \\
  \midrule
  null & -0.031 & 0.271 & 0.272 & -0.009 & 0.100 & 0.100 & -0.536 & 1.414 & 1.512 \\
  ph & -0.060 & 0.495 & 0.498 & 0.098 & 0.294 & 0.310 & -0.152 & 1.104 & 1.114 \\
  early & 0.309 & 0.580 & 0.657 & -0.088 & 0.471 & 0.480 & 0.153 & 0.531 & 0.552 \\
  tran & -0.011 & 0.085 & 0.086 & 0.035 & 0.139 & 0.143 & -0.033 & 0.052 & 0.061 \\
  cs & -0.008 & 0.125 & 0.125 & 0.033 & 0.194 & 0.197 & 0.018 & 0.127 & 0.128 \\
  msep & 0.053 & 0.185 & 0.193 & 0.164 & 0.289 & 0.332 & 0.036 & 0.118 & 0.124 \\
  delay\_1 & -0.026 & 0.382 & 0.383 & -0.206 & 0.263 & 0.334 & -0.326 & 0.890 & 0.948 \\
  delay\_2 & -0.036 & 0.386 & 0.387 & -0.215 & 0.241 & 0.323 & -0.819 & 1.010 & 1.300 \\
  delaycon & 0.095 & 0.382 & 0.393 & 0.215 & 0.253 & 0.332 & 0.071 & 0.668 & 0.672 \\
  \bottomrule
  \end{tabular}
  \label{tab:L_n1000}
\end{table}

\subsection{Additional Simulation Results: Impact of Penalization on \texttt{AdaRMST.ct}}

Figures \ref{fig:pen_300} to \ref{fig:pen_1000} illustrate the impact of penalization on power of \texttt{AdaRMST.ct} under sample sizes $n = 300, 600, 1000$ across different scenarios. The penalty parameter $c$ varies within $\{0, 0.0005, 0.001, 0.002, 0.005, 0.01\}$, and the initial restriction time is set to the default midpoint, $\tilde{L} = 2.2$.

\begin{figure}
  \centering
  \includegraphics[width=1\linewidth]{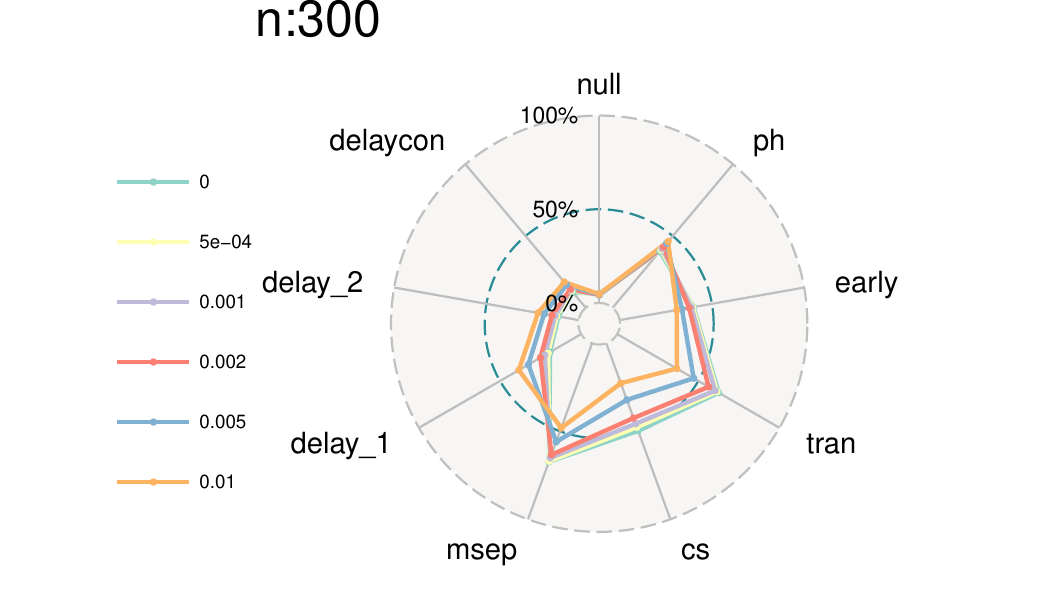}
  \caption{Impact of penalization on power of \texttt{AdaRMST.ct} for sample size $n = 300$.}
  \label{fig:pen_300}
\end{figure}

\begin{figure}
  \centering
  \includegraphics[width=1\linewidth]{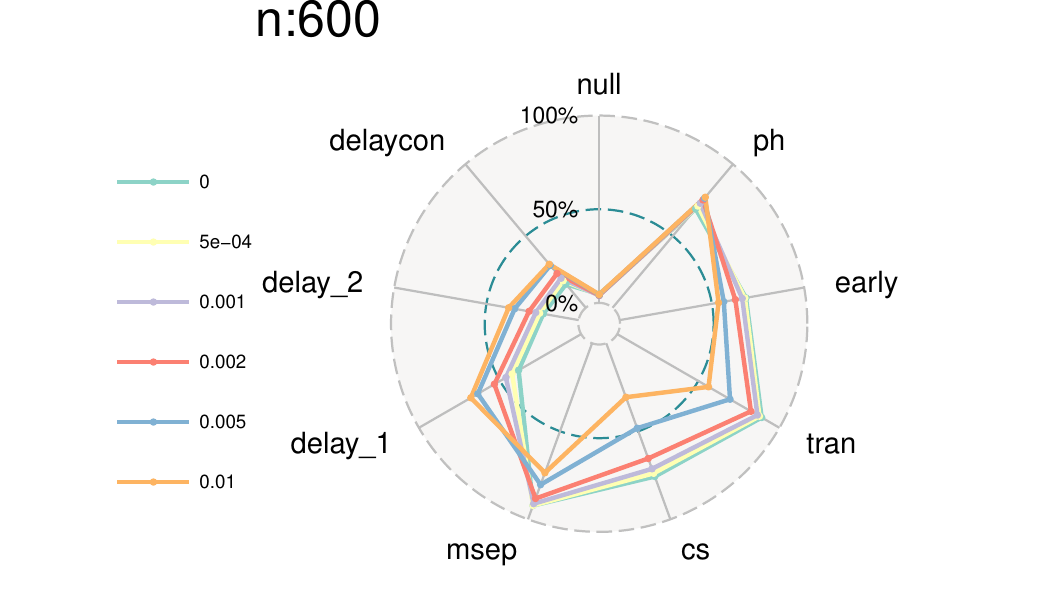}
  \caption{Impact of penalization on power of \texttt{AdaRMST.ct} for sample size $n = 600$.}
  \label{fig:pen_600}
\end{figure}

\begin{figure}
  \centering
  \includegraphics[width=1\linewidth]{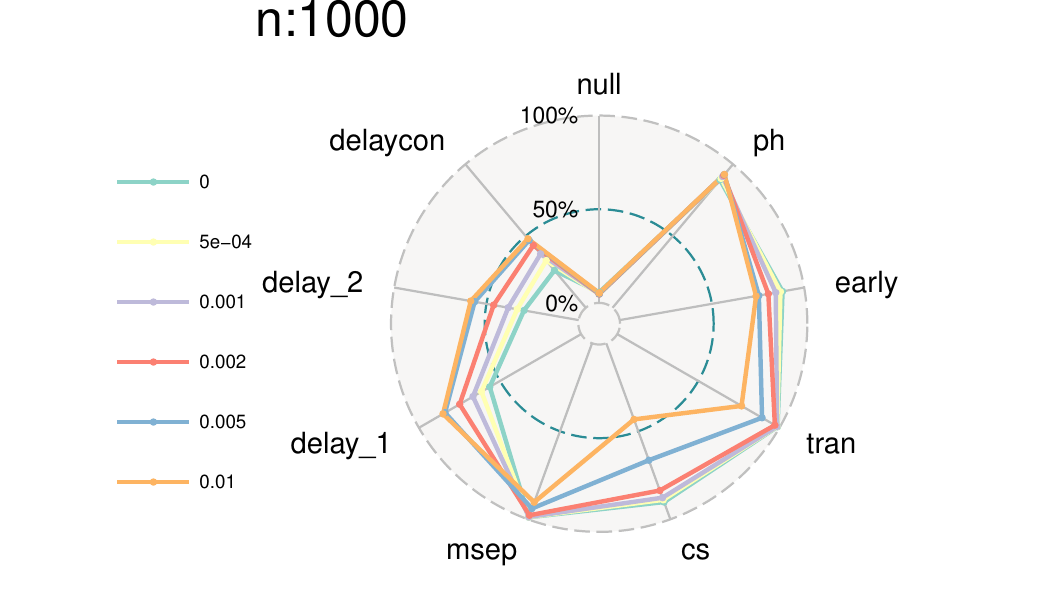}
  \caption{Impact of penalization on power of \texttt{AdaRMST.ct} for sample size $n = 1000$.}
  \label{fig:pen_1000}
\end{figure}

\end{document}